\documentclass[10pt]{article}
\usepackage{authblk}
\usepackage{cite}
\usepackage[cmex10]{mathtools}
\usepackage{amsfonts,amsthm,amssymb,bbm}
\usepackage{thmtools}
\usepackage{dsfont}
\usepackage{mymacros}
\usepackage{color}
\usepackage{datetime}
\usepackage{xspace}
\usepackage{graphicx}
\usepackage{siunitx}
\usepackage{enumitem}
\usepackage{bigints}
\usepackage{subcaption}
\usepackage{layouts}
\usepackage{mdframed}
\usepackage{tcolorbox}

\usepackage{xcolor}
\definecolor{base3}{HTML}{fdf6e3}

\usepackage{fourier}
\usepackage[T1]{fontenc}

\DeclareMathAlphabet{\mathcal}{OMS}{cmsy}{m}{n} 

\usepackage{hyperref}
\usepackage[capitalize]{cleveref}
\crefname{equation}{}{}
\crefname{enumi}{}{}

\usepackage{autonum}

\usepackage[T1]{fontenc}

\usepackage{tikz}
\usetikzlibrary{matrix,positioning,decorations.pathreplacing}

\usepackage{todonotes}
\usepackage[outdir=./figs/converted/]{epstopdf}
\DeclareGraphicsExtensions{.eps,.pdf}
\graphicspath{{figs/}}

\providecommand\given{}
\newcommand\SetSymbol[1][]{%
\nonscript\:#1\vert
\allowbreak
\nonscript\:
\mathopen{}}
\DeclarePairedDelimiterX\Set[1]\{\}{%
\renewcommand\given{\SetSymbol[\delimsize]}
#1
}

\usepackage{geometry}
\newcommand{\tf}[0]{k}
\newcommand{\kpint}[0]{\kappa}
\newcommand{\kp}[0]{\boldsymbol{\kappa}}
\newcommand{\qp}[0]{\mathbf{q}}

\newcommand{\rr}[0]{\mathbf{r}}
\newcommand{\ro}[0]{\rr_{0}}

\newcommand{\rp}[0]{\boldsymbol{\rho}}
\newcommand{\rpnb}[0]{\rho}
\newcommand{\rpo}[0]{\boldsymbol{\rho}_0}

\newcommand{\p}[0]{^\prime}
\newcommand{\na}[0]{\mathrm{NA}}
\newcommand{\kzo}[0]{k_z}
\newcommand{\kz}[1]{\kzo\left(#1\right)}

\newcommand{\ko}[0]{k}
\newcommand{\kmin}[0]{k_{\mathrm{min}}}
\newcommand{\kmax}[0]{k_{\mathrm{max}}}
\newcommand{\Nk}[0]{N_{k}}
\newcommand{\Deltak}[0]{\Delta_{k}}
\newcommand{\koi}[0]{k_{i}}

\newcommand{\Nf}{N_f}
\newcommand{\nf}{f}
\newcommand{\ns}{s}
\newcommand{\Ns}{N_s}
\newcommand{\zf}[0]{z_f}

\newcommand{\pspec}[0]{\zeta}

\newcommand{\lmin}[0]{\lambda_{\mathrm{min}}}
\newcommand{\lmax}[0]{\lambda_{\mathrm{max}}}
\newcommand{\Sh}[0]{\hat{S}}
\newcommand{\Ah}[0]{\hat{A}}

\newcommand{\Hbar}[0]{\bar{\mathbf{H}}}
\newcommand{\apfunc}[1]{\vartheta\left(#1 \right)}

\newcommand{\ones}[1]{\mathbbm{1}_{#1}}
\newcommand{\boldhat}[1]{\mathbf{\hat{#1}}}

\newcommand{\B}[0]{\boldhat{B}}
\newcommand{\Bbar}[0]{\bar{\mathbf{B}}}
\newcommand{\Abar}[0]{\bar{\mathbf{A}}}

\newcommand{\Abarl}[0]{\Abar^\qp}
\newcommand{\Bbarl}[0]{\Bbar^\qp}
\newcommand{\Ahl}[0]{\boldhat{A}^\qp}
\newcommand{\Bhl}[0]{\B^\qp}
\newcommand{\Vl}[0]{\mathbf{V}^\qp}
\newcommand{\phl}[0]{\boldhat{p}^\qp}
\newcommand{\pbl}[0]{\bar{\mathbf{p}}^\qp}
\newcommand{\sbl}[0]{\bar{\mathbf{s}}^\qp}
\newcommand{\Phitl}[0]{\tilde{\Phi}^\qp}
\newcommand{\Phil}[0]{\Phi^\qp}

\newcommand{\nset}[0]{\mathsf{N}^\qp}
\newcommand{\nsetbar}[0]{\bar{\mathsf{N}}^\qp}
\newcommand{\nsetperp}[0]{\left(\nset\right)^\perp}
\newcommand{\nsetbarperp}[0]{\left(\nsetbar\right)^\perp}

\newcommand{\bset}[0]{\mathsf{B}}
\newcommand{\vset}[0]{\mathsf{V}}

\newcommand{\krank}[1]{\mathrm{krank}\left(#1\right)}
\newcommand{\rank}[1]{\mathrm{rank}\left(#1\right)}
\newcommand{\dimm}[1]{\mathrm{dim}\left(#1\right)}
\newcommand{\nullspace}[1]{\mathrm{null}\left\{#1\right\}}

\newcommand{\rangespace}[1]{\mathrm{range}\left\{#1\right\}}

\renewcommand{\vect}[1]{\mathrm{vec}\left(#1 \right)}

\newcommand{\nsp}[0]{$N$-species }

\declaretheorem[name=Theorem, refname={Theorem, Theorems}]{theorem}
\declaretheorem[name=Lemma, refname={Lemma, Lemmas}]{lemma}

\declaretheorem[name=Corollary, refname={Corollary, Corollaries}]{corollary}

\declaretheoremstyle[
bodyfont=\normalfont,
]{defstyle}
\declaretheorem[style=defstyle,name=Definition, refname={Definition, Definitions}]{definition} 

\usepackage{pgf}
\usepackage{pgfplots}
\pgfplotsset{compat=newest}
\usepgfplotslibrary{groupplots}
\usepgfplotslibrary{dateplot}
\newcommand\inputpgf[2]{{
    \let\pgfimageWithoutPath\pgfimage
    \renewcommand{\pgfimage}[2][]{\pgfimageWithoutPath[##1]{#1/##2}}
    \input{#1/#2}
  }}

\begin{document}
\title{Composition-Aware Spectroscopic Tomography}
\date{}
\author[1]{Luke Pfister}
\author[1,2,3,4,5]{Rohit Bhargava}
\author[1]{Yoram Bresler}
\author[6]{P. Scott Carney}
\affil[1]{Department of Electrical and Computer Engineering, University of Illinois at Urbana-Champaign, Urbana, Illinois 61801, USA}
\affil[2]{Department of Bioengineering, University of Illinois at Urbana-Champaign, Urbana, Illinois 61801, USA}
\affil[3]{Department of Mechanical Science and Engineering, University of Illinois at Urbana-Champaign, Urbana, Illinois 61801, USA}
\affil[4]{Department of Chemical and Biomolecular Engineering, University of Illinois at Urbana-Champaign, Urbana, Illinois 61801, USA}
\affil[5]{Department of Chemistry, University of Illinois at Urbana-Champaign, Urbana, Illinois 61801, USA}
\affil[6]{The Institute of Optics, University of Rochester, Rochester, New York 14627, USA}

\maketitle
\begin{abstract}
  Chemical imaging provides information about the distribution of chemicals within a target. When
  combined with structural information about the target, in situ chemical imaging opens the door to
  applications ranging from tissue classification to industrial process monitoring. The combination
  of infrared spectroscopy and optical microscopy is a powerful tool for chemical imaging of thin
  targets. Unfortunately, extending this technique to targets with appreciable depth is
  prohibitively slow.

  We combine confocal microscopy and infrared spectroscopy to provide chemical imaging in three
  spatial dimensions. Interferometric measurements are acquired at a small number of focal depths,
  and images are formed by solving a regularized inverse scattering problem. A low-dimensional
  signal model is key to this approach: we assume the target comprises a finite number of distinct
  chemical species. We establish conditions on the constituent spectra and the number of
  measurements needed for unique recovery of the target. Simulations illustrate imaging of cellular
  phantoms and sub-wavelength targets from noisy measurements.
\end{abstract}

\section{Introduction}
\label{sec:intro}
Chemically specific imaging provides quantitative information about the distribution of chemicals
within a target. This may be accomplished through the use of exogenous chemicals or molecular
staining to improve contrast when the target is imaged with visible light. In many applications,
these dyes cannot be introduced \emph{in situ} and the agents are often
damaging to the target.

Vibrational spectroscopy with mid-infrared light presents a solution \cite{Bhargava2012}.
Absorption of mid-infrared light depends on chemical composition. The underlying chemistry of a
target can be determined, non-invasively, by illuminating the object with mid-infrared light and
recording an absorption spectrum.  

In principle, mid-infrared spectroscopy can provide chemically specific, spatially resolved imaging
in three spatial dimensions using a confocal scanning strategy: the target would be scanned
point-by-point in three spatial dimensions, and an absorption spectrum would be measured at each
point \cite{Davis2010a, Davis2010b}. For a target with two spatial dimensions, this is feasible- a
typical data set of $1024$ spectral samples over a $1024 \times 1024$ pixel grid requires on the
order an of hour of acquisition time and generates roughly 4 GB of data \cite{Reddy2013-High-Defin,
  Tiwari2016, Tiwari2017, Mittal2018-Simul-Cancer}. Scanning along a third spatial dimension (depth)
makes imaging even a single target impractical: the resulting dataset would require over 4 terabytes
of storage and roughly a month of acquisition time.

The key challenge in jointly measuring structural and chemical information is dimensionality: with
no constraints, the target can vary in three spatial and one spectral dimension. Existing imaging
modalities explicitly or implicitly rely on simple signal models to reduce the dimensionality of the target
and allow for practical imaging.

Optical Coherence Tomography (OCT) and Interferometric Synthetic Aperture Microscopy
(ISAM) are scattering-based imaging modalities that reconstruct the 3D spatial distribution of a
target by ignoring spectral variation \cite{Fercher2003, Davis2008, Ralston2007-Inter-Synth, Davis2007}
, although limited spectral information can be recovered at
the expense of spatial resolution by way of time-frequency analysis \cite{Oldenburg2007,Morgner2000,
  Bosschaart2013}.

Fourier Transform Infrared (FTIR) spectroscopy, a workhorse of academic and industrial labs
worldwide, neglects all spatial variation within the target---thus reducing the target to a single
dimension. An extension, FTIR microspectroscopy, provides spatially and spectrally resolved
measurements but requires the target to be very thin with only transverse heterogeneities. Unmodeled
spatial variations in the target cause scattering and diffraction, ultimately distorting the
measured spectra\cite{Davis2010a, Davis2010b}.




We propose an approach that bridges these two extremes and allows for practical, chemically specific
imaging. We call this \emph{spectroscopic tomography}. Rather than finely scanning the focus through
the axial dimension of the target, we acquire data at a small number of \emph{en-face} focal planes.
The target is recovered by solving the linearized scattering problem. A low-dimensional
model is used to regularize the inverse problem: we model the target as the linear combination of a
finite number of distinct chemical species.
We call this the $N$-species approximation. We develop
a set of algebraic conditions for unique recovery and examine the conditioning of the inverse
problem. Reconstructions from synthetic phantom data illustrate the promise of the model.

The \nsp model was investigated for a one-dimensional target in
\cite{Deutsch2015}. However, \cite{Deutsch2015} involved several involve several
unrealistic assumptions, leading to results of unrealistically high quality. We
extend this work in several directions: we (i) use a non-asymptotic forward
model; (ii) demonstrate material-resolved reconstruction of samples with two
spatial dimensions (one transverse and depth, easily extended to three spatial
dimensions) from synthetic scattering data that are not generated according to the first Born
approximation; and (iii) refine the conditions for recovery of a sample
consisting of \nsp from interferometric scattering experiments.

The paper is organized as follows. In \cref{sec:forward} we describe the forward model.
\cref{sec:nspecies} describes the \nsp model in greater detail. We discuss the sampling and
discretization procedure in \cref{sec:sampling_and_discretization}. We investigate the inverse
problem in \cref{sec:nspecies_inverse}, and demonstrate the method by performing numerical
reconstructions from simulated measurements in \cref{sec:nspecies_sims}.

\subsection{Notation}
We write the set of integers $\left\{0, 1, \hdots, N-1\right\}$ as $[N]$ and the two-fold Cartesian
product $[N] \times [N]$ as $[N]^2$. We write the imaginary unit as $\mathrm{i}$. Finite-dimensional vectors
are denoted by lower-case bold letters, \eg $\mathbf{x} \in \Cbb^N$. Finite-dimensional matrices and
multi-dimensional arrays are written using upper-case bold letters. We adopt Matlab-style indexing
notation: given a matrix $\mathbf{A}\in\Cbb^{N \times M}$, its $i$-th row is $\mathbf{A}[i,
\colon]$, the $j$-th column is $\mathbf{A}[\colon, j]$, and the $(i, j)$-th element is $\mathbf{A}[i,j]$.
We denote the vector $\vect{\mathbf{A}} \in \Cbb^{NM}$ is formed by stacking the columns of
$\mathbf{A}$ into a single vector (\ie, row-major ordering). The range, null space, and rank of a
matrix $\mathbf{A}$ are written $\rangespace{\mathbf{A}}, \nullspace{\mathbf{A}}$, and
$\rank{\mathbf{A}}$. Given $\mathbf{x} \in \Cbb^N$, the diagonal matrix
$\mathrm{diag}\left\{\mathbf{x}\right\} \in \Cbb^{N \times N}$ has the entries of $\mathbf{x}$ along
its main diagonal.

The transpose (\emph{resp.} Hermitian transpose) of a matrix is written $\mathbf{A}^{\mathsf{T}}$
(\emph{resp.} $\mathbf{A}^{\mathsf{H}}$). The
$\ell_p$ norm of $\mathbf{x} \in \Cbb^N$ is $\norm{\mathbf{x}}_p = \left(\sum_{j=1}^N
  \abs{\mathbf{x}[j]}^p\right)^{1/p}$. For vectors in $\Rbb^2$ or $\Rbb^3$ we use the shorthand
$\abs{r} = \norm{\rr}_2$. The $N \times N$ identity matrix is $\mathbf{I}_N$, and the vector $[1, 1,
\hdots 1]^{\mathsf{T}} \in \Rbb^N$ is written $\ones{N}$. The tensor (or Kronecker) product between
matrices $\mathbf{A}$ and $\mathbf{B}$ is $\mathbf{A} \otimes \mathbf{B}$.

\section{Preliminaries}
\label{sec:forward}

We characterize the sample under investigation by its complex refractive index,
$n(\rr, \ko) = n_b + \delta n(\rr, \ko)$ where $n_b$ is the refractive index of
the background medium and $\delta n$ is the perturbation due to the sample; for
simplicity, we take $n_b = 1$. Here, $\rr = (x, y, z) = (\rp, z)$, where $\rp$
are the transverse dimensions and $z$ indicates the axial dimension. We assume
that $\delta n$ is (spatially) supported in the bounded region $\Gamma \subset
\Rbb^3$. The free-space wavenumber $\ko$ is related to temporal frequency
$\omega$ by $\ko = \omega / c$, where $c$ is the speed of light in free space.
The real part of the complex refractive index is the ratio between $c$ and the
phase velocity in the medium, while the imaginary part indicates attenuation due
to propagation through the target.

Under the first Born approximation, the obtained measurements are linear in the complex
susceptibility $\eta \triangleq n^2 - 1$;  we will work with the susceptibility
rather than the refractive index. Note that $\eta$ is also supported on $\Gamma$.

In the context of spectroscopy, the ``spectrum'' of a sample usually refers either to its complex
refractive index or only the imaginary part of the refractive index. Consider a homogeneous medium
with refractive index $n(\ko) = n_r(\ko) + \mathsf{i} \kappa(\ko)$. The real part, $n_r(\ko)$, has
mean value greater than one and the imaginary part, $\kappa(\ko)$, is non-negative. Relating
$\eta(\ko)$ to $n(\ko)$, we have
\begin{equation}
\eta(\ko) = n(\ko)^2 - 1 = n_r(\ko)^2 - \kappa(\ko)^2 - 1 + 2 \mathsf{i} n_r(\ko) \kappa(\ko).
\end{equation}
Unlike the refractive index, the mean value of the real part of $\eta(\ko)$
may be less than one and can be negative. The imaginary
part of $\eta(\ko)$ remains non-negative.

\subsection{Interferometric Synthetic Aperture Microscopy}
\begin{figure}[t]
  \centering
  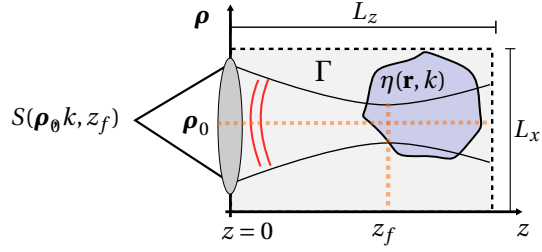
  \caption{Geometry and notation for scattering problem. The illuminating aperture is located at
    $(\rpo, 0)$. The field emerges from the aperture and is focused to the plane $z = \zf$. The
    incident beam interacts with the sample, $\eta$, and the backscattered light (red) is collected
    through the aperture to produce the measurement $S(\rpo, \ko, \zf)$. }
  \label{fig:beam_geometry} 
\end{figure}
In this section, we review the forward model relating the target, $\eta$, to the observed data.
For a complete derivation, see \cite{Ralston2006, Davis2008, Davis2007, Sentenac2018}.

The imaging geometry is depicted in \cref{fig:beam_geometry}. An aperture located at $(\rpo, 0)$
emits a broadband Gaussian beam focused to a point $\ro = (\rpo, \zf)$ within the sample. The
illuminating field interacts with the sample, and a portion of the light is scattered backwards and
is collected through the aperture. The aperture is raster scanned along the transverse coordinates
$\rpo$. At each point the scattered field is measured interferometrically, from which we use
standard techniques to recover the complex (phase-resolved) measurements. In the remainder of this
paper, we ignore the interferometric aspects of data acquisition and work directly with the
phase-resolved measurements.

Under the first Born approximation, the measured data $S(\rpo, \ko, \zf)$ are a linear function of
$\eta$; we have
\begin{align}
S(\rpo, \ko, \zf) &= \iint A(\rpo - \rp, z -\zf, \ko) \eta(\rp, z, \ko) \ \mathrm{d}z \ \mathrm{d}^2 \rpnb ,
                    \label{eq:non_asymp_space}
\end{align}
or, after taking a Fourier transform with respect to the scanning dimension $\rpo$, 
\begin{align}
  \Sh(\kp, \ko, \zf) = \frac{1}{2 \pi} \int S(\rpo, \ko, \zf) e^{-\mathrm{i} \kp \cdot \rpo} \mathrm{d}^2 \rpnb  
                   = \int \Ah(\kp, z - \zf, \ko) \hat{\eta}(\kp, z, \ko) \mathrm{d}z .
  \label{eq:non_asymp}
\end{align}
We call the function $\hat{A}$ the \emph{ISAM kernel}. This function is itself defined by an
integral; explicitly,
\begin{equation}
  \Ah(\kp, z, \ko) \triangleq \frac{\abs{\pspec(\ko)}^2}{\ko^2 \na^2}
  \bigintsss_{\Omega(\kp, \ko)} 
  \frac{
  \exp{\left\{-\frac{1}{(\ko \na)^2}\left( \abs{\kp\p}^2 + \abs{\kp - \kp\p}^2   \right)
      + \mathrm{i} z \left(\kz{\kp\p, \ko} + \kz{\kp - \kp\p, \ko} \right)\right\}}
  }{\kz{\kp\p, \ko}}
  \ \mathrm{d}^2 \kpint\p,
  \label{eq:isam_kernel_integral}
\end{equation}
where $\kz{\kp, \ko} \triangleq \sqrt{\ko^2 - \abs{\kp}^2}$
and the set
$\Omega(\kp, \ko) \triangleq \left\{ \kp\p \in \Rbb^2 : \abs{\kp -\kp\p} \leq \ko , \ \abs{\kp\p}
  \leq \ko \right\} \subset \Rbb^3$
restricts the integral to propagating modes.  
The scalar $\na > 0$ is the numerical aperture of the illumination lens and 
$\abs{\pspec(\ko)}^2$ is the power spectrum of the illumination source. We assume that
$\pspec(\ko)$ is supported on the interval $[\kmin, \kmax]$.


\subsection{Image Reconstruction using ISAM}
\label{sec:asymptotics}
\begin{figure}
  \centering
  \fontsize{8pt}{11pt}\selectfont
  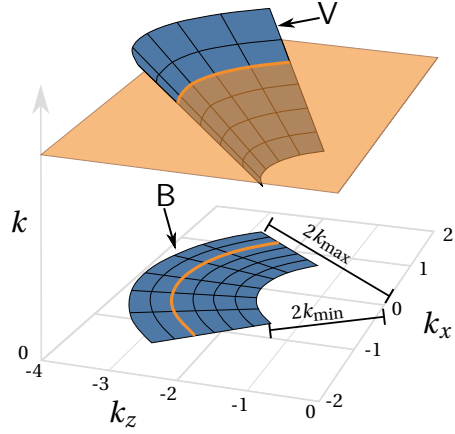
  \caption{Observable Fourier components for a target with two spatial and one spectral dimensions.
    The intersection of $\vset$ with a plane of constant $\ko$ becomes an arc of constant
    radius when projected onto the $(k_x, \kzo)$ plane.
  }
  \label{fig:sampling_3d}
\end{figure}
Next, we discuss recovering the object $\eta$ from measurements of the form \cref{eq:non_asymp}.
First, note that $\pspec(\ko)$ in \cref{eq:isam_kernel_integral} ensures that $\hat{A}(\kp, z, \ko)$
vanishes for any $\ko \notin [\kmin, \kmax]$. Further, $\Omega(\kp, \ko)$ is empty for $\abs{\kp} >
2 \ko$ and so $\hat{A}(\kp, z, \ko)$ vanishes for all $\abs{\kp} > 2 \kmax$. Thus the measurements
are related to the bandlimited transverse Fourier transform of the object.

Previous derivations of ISAM continue by invoking a pair of approximations to the integral
\cref{eq:isam_kernel_integral}.  One approximation holds when $\ko\abs{z - \zf}$ is small
and the other holds when the same quantity is large.  Both approximations are of the form
\begin{equation} 
  \label{eq:asymp}
  \Ah(\kp, z-\zf, \ko) \approx \chi\left(\kp, 2 \ko \right) \abs{\pspec(\ko)}^2 \apfunc{\kp, \ko}
  \upsilon(z - \zf) e^{\mathrm{i} \kz{\kp, 2 \ko} (z - \zf)}
\end{equation}
where
\begin{equation}
  \chi\left(\kp, \ko \right) \triangleq \begin{cases}
    1, & \abs{\kp} \leq \ko \\ 
    0, & \text{otherwise, }
  \end{cases}
\end{equation}
the function $\apfunc{\kp, \ko}$ captures the transverse bandpass nature of the imaging system due
to the aperture, and $\upsilon(z)$ is a depth-dependent weighting function. The precise form of these
functions depends on if $\ko\abs{z - \zf}$ is large or small; in either case,  $\apfunc{\kp, \ko} \propto
e^{-\frac{\abs{\kp}^2}{(\ko \na)^2}}$ and $\upsilon(z)$ falls off as $z^{-1}$ \cite{Davis2007}.

Inserting \cref{eq:asymp} into the measurement model \cref{eq:non_asymp}, we have
\begin{equation}
  \Sh(\kp, \ko, \zf) \approx \chi\left(\kp, 2 \ko \right) \abs{\pspec(\ko)}^2 \apfunc{\kp, \ko}
  e^{\mathrm{i} \kz{\kp, 2\ko}\zf } 
  \int \upsilon(z - \zf) \hat{\eta}(\kp, z, \ko)
  e^{-\mathrm{i} \kz{\kp, 2\ko} z} \mathrm{d}z \label{eq:asymp_integral}.
\end{equation}
Consider a single, fixed, focal plane; this is the usual setting for ISAM imaging.  Define the weighted susceptibility
\begin{equation}
  \hat{\xi}_{\zf}(\kp, z, \ko) \triangleq \upsilon(z - \zf) \hat{\eta}(\kp, z, \ko).
\end{equation}
The integral in \cref{eq:asymp_integral}
is the Fourier transform of $\hat{\xi}_{\zf}$ with respect to $z$ evaluated at the
frequency $-\kz{\kp, 2\ko}$;  thus
\begin{align}
  \label{eq:asymp_diagonal}
  \Sh(\kp, \ko, \zf) &\approx \chi\left(\kp, 2 \ko \right)\abs{\pspec(\ko)}^2 \apfunc{\kp, \ko} e^{\mathrm{i}
                       \kz{\kp, 2\ko}\zf}
    \
    \hat{\hat{\xi}}_{\zf}(\kp, -\kz{\kp,  2 \ko},  \ko),
\end{align}
where the double hat indicates a 3D Fourier transform with respect to $\rr = (x, y, z)$. This is
a generalized projection-slice theorem: the ISAM data are approximately the bandlimited (spatial) Fourier
transform of the weighted susceptibility evaluated on a three dimensional
surface that is parameterized by $\kp$ and $\ko$. By varying $\kp$ and $\ko$, we are able to observe a
curved 3D ``slice'' of the four-dimensional function $\hat{\hat{\xi}}_{\zf}(\kp, \kzo, \ko)$ constrained to the
surface
\begin{equation}
  \vset \triangleq \left\{ (k_x, k_y, \kzo, \ko) :
      \sqrt{k_x^2 +k_y^2 + \kzo^2} = 2 \ko, \  \kzo < 0,  \ k_x^2 + k_y^2 \leq 4 (\ko \na)^2, \ \kmin \leq \ko \leq \kmax
  \right\}.
\end{equation} 
The sampling surface for a target with two spatial dimensions, \ie $\rr = (x, z)$, is illustrated in
\cref{fig:sampling_3d}; that we can only observe $\kzo < 0$ is due to the backscattering
geometry. As defined, $\vset$ contains only the Fourier components above the $e^{-2}$ cutoff
frequency of $\apfunc{\kp, \ko}$. This is arbitrary as $\apfunc{\kp, \ko}$ decays smoothly.

We cannot recover an arbitrary object given ISAM data at a single focal plane. If the focus $\zf$
were scanned in addition to $\rpo$, we could further simplify by taking a Fourier transform along
$\zf$. Then, the measurements would be of the form $\hat{\Sh}(\kp, \kzo, \ko) = \hat{\Ah}(\kp, \kzo, \ko)
\hat{\hat{\eta}}(\kp, \kzo, \ko)$, where the double hat indicates the 3D Fourier transform with
respect to $\rr$. Now, $\eta$ could be recovered using a standard deconvolution procedure.
Unfortunately, this is infeasible for reasons described in \cref{sec:intro}.

The situation is simplified if $\eta$ is not a function of $\ko$; such an object is said to be
\emph{non-dispersive}. This is one of the key assumptions on
which ISAM, OCT, diffraction tomography, and reflection tomography are built
\cite{Fercher2003,Wu1987-Diffr-Tomog,Kak2001}. In this case, the measurements are related to a 3D
slice of the 3D target $\eta(x,y,z)$.  The observable Fourier components are
\begin{align}
  \bset &\triangleq
               \left\{ (k_x, k_y, \kzo) :
      \sqrt{k_x^2 +k_y^2 + \kzo^2} = 2 \ko, \  \kzo < 0,  \ k_x^2 + k_y^2 \leq 4 (\ko \na)^2, \ \kmin \leq \ko \leq \kmax
               \right\}
\end{align}
The region $\bset$ is called the \emph{optical passband} of the ISAM imaging system. Strictly
speaking, we observe the Fourier components of the \emph{weighted} susceptibility on $\bset$, but
this distinction is usually ignored. Only a non-dispersive (weighted) object whose spatial Fourier
transform is supported on $\bset$ can be perfectly imaged by the ISAM system with a single focal
plane. Otherwise, ISAM is able to recover, at best, a spatial bandpass version of the original
target. In the visualization of \cref{fig:sampling_3d}, $\bset$ is the ``shadow'' cast by $\vset$ onto
the plane $\ko = 0$.

We do not directly use the approximate kernel \cref{eq:asymp} in this paper. However, we use the
insight provided by this approximation as a guide; in particular, the Fourier transform
interpretation and the optical passband $\bset$ inform the sampling procedure and help 
establish fundamental limits of the imaging system.

\section{The \nsp Model}
\label{sec:nspecies}
\subsection{The Model}
\label{sub:nspecies}
Existing imaging modalities use simplified signal models to reduce the dimensionality of the sample
and allow for practical imaging. ISAM, optical coherence tomography, diffraction tomography, and
reflection tomography require the target be either non-dispersive or have known spatially invariant
dispersion characteristics. In this case, the susceptibility is of the form $\eta(\rr, \ko) = p(\rr)
h(\ko)$, where $p(\rr)$ captures the spatial density of the target and $h(\ko)$ characterizes the
wavelength-dependent dispersion characteristics. If $h(\ko)$ is known, only $p(\rr)$ must be
determined---thus reducing the problem to recovery of a three-dimensional object. Conversely,
Fourier Transform Infrared spectroscopy of a bulk medium assumes that the sample is spatially
homogeneous, so that $\eta(\rr, \ko) = h(\ko)$. An extension, FTIR microscopy, models the sample as
a thin absorbing screen; thus $\eta(\rr, \ko) = \eta(\rp, \ko)$, a three-dimensional object.

These examples severely restrict the class of samples that can be imaged. We propose a model
that is more expressive than these examples while still allowing practical imaging.
\begin{definition}[The \nsp model \cite{Deutsch2015}]
  An object, described by a susceptibility $\eta(\rr, \ko)$,
  is said to satisfy the \nsp model if 
\begin{equation}
  \label{eq:nspecies_separable}
  \eta(\rr, \ko) = \sum_{\ns=1}^{\Ns} p_{\ns}(\rr) h_{\ns}(\ko).
\end{equation}
The function $p_{\ns}(\rr)$ captures the spatial variation of the $\ns$-th species and is called the
\emph{spatial density}. If species $\ns$ is not present at location $\rr$, then $p_{\ns}(\rr) = 0$. The
complex function $h_{\ns}$ models the wavelength-dependent properties of the $\ns$-th species and is
called the \emph{spectral profile}.
\end{definition}
The \nsp model, introduced in \cite{Deutsch2015}, is a rank $\Ns$
approximation to a general susceptibility.  
A similar decomposition has been applied to magnetic resonance
spectroscopic imaging, where  it is called the Partially Separable (PS) function model
\cite{Liang2007, Nguyen2010, Nguyen2011, Nguyen2013}, and to material
decomposition in X-ray tomography \cite{Alvarez1976-Energ-Selec, Cammin2012-Spect}.

\subsection{Spectroscopic Tomography with the $N$-Species Model}
Inserting the \nsp model \cref{eq:nspecies_separable} into the linearized forward model
\cref{eq:non_asymp}, we have 
\begin{equation}
  \label{eq:nspecies_meas}
  \Sh(\kp, \ko, \zf) = \sum_{\ns=1}^{\Ns} h_{\ns}(\ko) \int_{-\infty}^\infty \Ah(\kp, z - \zf, \ko) \hat{p}_{\ns}(\kp, z) \mathrm{d}z.
\end{equation}
At a given focal plane, the measurements are the sum of $\Ns$ independent ISAM experiments, each on
a non-dispersive object $\hat{p}_{\ns}(\kp, z)$ and each weighted by the spectral profile $h_{\ns}(\ko)$. In
what follows, we study inverse problem associated with spectroscopic optical tomography: we wish to
recover an object that satisfies the \nsp model from measurements of the form
\cref{eq:nspecies_meas}.

In the single species case, the inverse problem can be solved from data acquired at
single focal plane---this is the usual ISAM problem. On the other hand, an arbitrary sample can be
recovered by finely scanning in all three spatial dimensions (\ie, along $\rpo$ and $\zf$) and
acquiring a spectrum at each point, but this is infeasible as described in \cref{sec:intro}.

The \nsp model is a middle ground between a single species object and an arbitrary
one.  We show that the number of measurements required to solve the inverse problem also
lies in a middle ground between these two extremes:  in particular, an object
satisfying the \nsp model can be recovered using $\Nf \approx \Ns$ focal planes.

We divide the inverse problem into three distinct cases.
\begin{enumerate}[label=(P\arabic*)] 
\item \emph{Known Spectra.} Assume the spectral profiles $\left\{h_{\ns}\right\}_{\ns=1}^{\Ns}$
  are fixed and known. The task reduces to a linear inverse problem---recovery of the
  $\left\{\hat{p}_{\ns}\right\}_{\ns=1}^{\Ns}$ from measurements of the form
  \cref{eq:nspecies_meas}. \label{enum:linear}
\item \emph{Spectra from a Dictionary.} Assume the target comprises at most $\Ns$ chemical species,
  but the spectral profiles are drawn from a (known) dictionary of some $M_s > \Ns$
  possible spectra. The inverse problem can be phrased as either a linear inverse problem over the
  entire dictionary, or as a nonlinear problem where the solution is constrained to lie in a union
  of subspaces.  
  \label{enum:dictionary}
\item \emph{Fully Blind.} Both the $\left\{h_{\ns}\right\}_{\ns=1}^{\Ns}$ and
  $\left\{\hat{p}_{\ns}\right\}_{\ns=1}^{\Ns}$ are unknown and must be recovered from measurements of the form
  \cref{eq:nspecies_meas}. This is a \emph{bilinear} inverse problem in $h_{\ns}$ and
  $p_{\ns}$.
  \label{enum:fully_blind}
\end{enumerate}
In this paper, we limit our attention to cases \cref{enum:linear,enum:dictionary}. Our analysis is
based on a discretized form of \cref{eq:nspecies_meas} wherein all quantities are replaced by
finite-dimensional versions, resulting in a so-called ``discrete-to-discrete'' inverse problem
\cite{Kress2014-Linear-Integ-Equat, Barrett2004-Foundations}. Next, we describe the sampling and
discretization procedure.

\section{Sampling and Discretization of the Forward Model}
\label{sec:sampling_and_discretization}
\subsection{Sampling}
\label{sec:sampling}
The instrument acquires samples of the spatial-domain measurement equation
\cref{eq:non_asymp_space}. We assume the object is (spatially) supported in a region $\Gamma \subset \Rbb^3$;
here, we take $\Gamma= [0, L_x] \times [0, L_y] \times [0, L_z]$. We write the number of samples as
$N_i$ and the discretization or sampling interval as $\Delta_i$ for $i\in \left\{x,y,z, k\right\}$. We obtain
measurements at the transverse aperture locations $\rpo = (n_x \Delta_x, n_y \Delta_y)$ for integers
$n_x, n_y$. The parameters are chosen to cover $\Gamma$, \ie $N_i \Delta_i = L_i$ holds for $i=x, y,
z$. For simplicity, we assume the sampling parameters are the same along the $x$ and $y$ directions:
$N_x = N_y$, $\Delta_x = \Delta_y$, and $L_x = L_y = N_x \Delta_x$. The wavenumber is sampled
uniformly over the interval $[\kmin, \kmax]$ with sampling interval $\Deltak$; the $i$-th measurement
wavenumber is $\koi \triangleq \kmin + i \Deltak$. We acquire data at $\Nf$ focal planes, written
$\left\{z_i : i=1,2,\hdots \Nf \right\}$. The same sampling parameters are used at each focal
plane; in particular, the set of sampled wavenumbers does not change.

We choose the sampling parameters as we would for a standard, single-species ISAM
problem.The necessary sampling intervals can be motivated using the approximate forward model
\cref{eq:asymp}. Under this model, it can be shown that ``point spread function'' $\abs{A(\rp, \ko,
  z)}$ (approximately) decays like a Gaussian in $\abs{\rp}$;  the measurements are ``essentially''
space limited \cite{Landau1962-Prolat-Spher}. We take $L_x$ large enough
to safely neglect the unmeasured data. Moreover, for fixed $\zf$, the measurements
are bandlimited in $\qp$ to $[-\kmax \sin{\na}, \kmax \sin{\na}]$; we sample along the transverse
dimension at intervals $\Delta_x, \Delta_y < \pi / (\kmax \sin{\na})$.
\subsection{Discretization}
Given samples of the measurements \cref{eq:non_asymp_space}, we take the 2D Discrete Fourier Transform (DFT) with
respect to the transverse coordinates and write the result as the array $\boldhat{S} \in \Cbb^{N_x
  \times N_x \times \Nk \times \Nf}$. For simplicity, we take $N_x$ to be even.
The 2D-DFT coordinate is written $\qp \in [N_x]^2$ .

As the measurements are bandlimited and the object is compactly supported, the DFT of the
measurements can be well-approximated by samples of the Fourier transform of the continuous model
\cref{eq:non_asymp} \cite{Barrett2004-Foundations}. The discretized \nsp measurement model is
\begin{equation}
  \boldhat{S}[\qp, m, \nf] = \sum_{\ns=1}^{\Ns} \mathbf{h}_{\ns}[m] \sum_{n=0}^{N_z-1} \boldhat{A}_{\nf}[\qp, m,  n] \boldhat{P}_{\ns}[\qp, n],
  \label{eq:disc_tensor_form}
\end{equation}
where $\mathbf{h}_{\ns} \in \Cbb^{\Nk}$ and $\boldhat{P}_{\ns} \in \Cbb^{N_x \times N_x \times N_z}$
are the discretized spectral profile and the 2D-DFT of the $\ns$-th spatial density with respect to
the transverse coordinates.
The coefficients $\boldhat{A}_{\nf}$ are obtained by sampling the continuous ISAM kernel
\cref{eq:isam_kernel_integral}; that is,
\begin{equation}
  \boldhat{A}_{\nf}[\qp, m, n] \triangleq \hat{A}\left(\kp(\qp), \kmin + m \Deltak, n \Delta_z - z_{\nf}\right).
\label{eq:isam_matrix_coeff}
\end{equation}
Here, $\kp(\qp) \triangleq (k_x(q_x), k_y(q_y))$ relates the DFT index $\qp$ to the continuous Fourier
coordinate $\kp$, where
\begin{equation}
  k_x(q_x) = \begin{cases}
    2 \pi q_x / L_x \quad & q_x < N_x / 2 \\
    2 \pi (q_x - N_x) / L_x \quad & \text{otherwise,}
  \end{cases}
  \label{eq:dft_to_ft_coord}
\end{equation}
and the same holds for $q_y$ and $k_y$.  


\begin{figure}
  \centering
  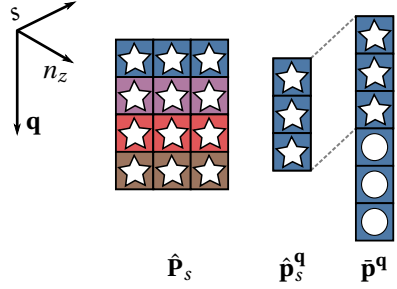
  \caption{
    The various unfoldings of the discretized spatial densities with  $\Ns = 2$.  Here, block color
    indicates the value of $\qp$.  Species 1 is marked with a star, while species 2 is indicated
    with a circle. }
  \label{fig:density_slicing}
\end{figure}

\subsection{Block-Matrix Form of $N$-Species Forward Model} With the spectral profiles fixed, the
measurements $\boldhat{S}$ are a linear function of the spatial densities. We can write the forward
model \cref{eq:disc_tensor_form} as a matrix-vector product, where the vector depends only on the
spatial densities. Moreover, as the forward model is separable in $\qp$, the resulting matrix has a
block-diagonal structure. In what follows, we will analyze the inverse problem for each $\qp$
independently.

We write the forward model \cref{eq:disc_tensor_form} as a separate linear system for each $\qp$.
Recall upper-case bold letters refer to matrices or arrays and
lower-case bold letters refer to vectors. We use a bar to denote objects that have been ``stacked''
or vectorized. Subscripts are used to slice a array with respect to the last index: \eg,
$\boldhat{S}_{\nf}$ represents all measurements from the $\nf$-th focal plane, while
$\mathbf{h}_{\ns}$ is the spectral profile for the $\ns$-th species. A superscript indicates a
submatrix or vector formed for particular value of $\qp$. In particular, we define
\begin{align}
  \boldhat{s}_{\nf}^{\qp} &\triangleq \boldhat{S}[\qp, \colon,  \nf] \in \Cbb^{\Nk}
\\
  \boldhat{A}^{\qp}_{\nf} &\triangleq \boldhat{A}_{\nf}[\qp, \colon, \colon]\in \Cbb^{\Nk \times N_z} \\
  \boldhat{p}_{\ns}^{\qp} &\triangleq \boldhat{P}_s[\qp, \colon] \in \Cbb^{N_z}.
\end{align}
Further, define the diagonal matrix $\mathbf{D}_{\ns} \triangleq \diag{\left( \mathbf{h}_{\ns} \right)} 
\in \Cbb^{\Nk \times \Nk}$.
Now, for fixed $\qp$ and $\nf$, 
\cref{eq:disc_tensor_form} can be written
\begin{equation}
  \boldhat{s}_{\nf}^{\qp}  = \sum_{{\ns}=1}^{\Ns} \mathbf{D}_{\ns} \boldhat{A}_{\nf}^{\qp} \boldhat{p}_{\ns}^{\qp}.
  \label{eq:disc_sum_form}
\end{equation}
The collection of \cref{eq:disc_sum_form} for $\nf \in [\Nf]$ can be written as a single linear
system in block form as
  \begin{equation}
    \begin{bmatrix}
      \boldhat{s}^{\qp}_1 \\
      \vdots \\
      \boldhat{s}^{\qp}_{\Nf}
    \end{bmatrix} = 
    \begin{bmatrix}
      \mathbf{D}_1 \boldhat{A}_{1}^\qp   &   \hdots & \mathbf{D}_{\Ns} \boldhat{A}_{1}^\qp    \\
      \vdots &  \ddots & \vdots \\
      \mathbf{D}_1 \boldhat{A}_{\Nf}^\qp   &   \hdots & \mathbf{D}_{\Ns} \boldhat{A}_{\Nf}^\qp    \\
    \end{bmatrix}
    \begin{bmatrix}
      \boldhat{p}^{\qp}_1 \\
      \vdots \\
      \boldhat{p}^{\qp}_{\Ns}
    \end{bmatrix},
  \end{equation}
or concisely as 
\begin{equation}
  \mathbf{\bar{s}}^\qp = \Phi^\qp \bar{\mathbf{p}}^\qp
  \label{eq:nsp_concise_one_block}
\end{equation}
where the vectors
\begin{align}
  \bar{\mathbf{p}}^\qp &\triangleq  \vect{\boldhat{P}[\qp, \colon, \colon]}
                       = [(\boldhat{p}^\qp_{1})^{\mathsf{T}}, \hdots, (\boldhat{p}^\qp_{\Ns})^{\mathsf{T}}]^{\mathsf{T}}
                       \in \Cbb^{\Ns N_z} \label{eq:pbar_l_defn}\\
  \mathbf{\bar{s}}^\qp &\triangleq
                       \vect{\boldhat{S}[\qp, \colon, \colon]} = 
                       [(\boldhat{s}^\qp_{1})^{\mathsf{T}}, \hdots, (\boldhat{s}^\qp_{\Nf})^{\mathsf{T}}]^{\mathsf{T}} \in \Cbb^{\Nf \Nk}, 
\end{align}
contain the spatial densities for each species and measurement for all focal planes, respectively, and 
\begin{equation}
  \Phi^\qp \triangleq
  \begin{bmatrix}
    \mathbf{D}_1 \boldhat{A}_{1}^\qp   &   \hdots & \mathbf{D}_{\Ns} \boldhat{A}_{1}^\qp    \\
    \vdots &  \ddots & \vdots \\
    \mathbf{D}_1 \boldhat{A}_{\Nf}^\qp  &   \hdots & \mathbf{D}_{\Ns} \boldhat{A}_{\Nf}^\qp    \\
  \end{bmatrix} \in
  \Cbb^{\Nf \Nk \times \Ns N_z}.
  \label{eq:phi_subblock}
\end{equation}
Each block-row of $\Phi^\qp$ corresponds to a single transverse Fourier component of measurements
taken at a single focal plane, and the $\ns$-th block-column corresponds to the $\ns$-th species.


\subsection{Construction using Khatri-Rao product}
\label{sec:khatri-rao}
We briefly discuss an alternative construction of $\Phil$ that
connects the \nsp inverse problem to a broad range of related problems, which we explore
in \cref{sec:related_kr}. 
\begin{definition}
  The \emph{row-wise Khatri-Rao} product of matrices $\mathbf{A} \in \Cbb^{m \times n_1}$
  and $\mathbf{B} \in \Cbb^{m \times n_2}$ is
  \begin{equation}
    \mathbf{A} \odot \mathbf{B} =
    \begin{bmatrix}
      \mathbf{A}[1,\colon] \otimes \mathbf{B}[1, \colon] \\
      \vdots \\
      \mathbf{A}[m,\colon] \otimes \mathbf{B}[m, \colon] 
    \end{bmatrix} \in \Cbb^{m \times n_1 n_2},
  \end{equation}
  \ie each row of $\mathbf{A} \odot \mathbf{B}$ is the Kronecker product of the corresponding rows of
  $\mathbf{A}$ and $\mathbf{B}$.
\end{definition}
We use the Khatri-Rao product to construct $\Phi^\qp$. 
First, we define
the matrix of spectral profiles $\mathbf{H} \in \Cbb^{\Nk \times \Ns}$ by
\begin{equation}
  \mathbf{H}[m, \ns] = \mathbf{h}_{\ns}[m].
\end{equation}
Now, the first block-row of $\Phi^\qp$ is $\mathbf{H} \odot \hat{\mathbf{A}}_1^\qp$. To obtain all block-rows of
$\Phi^\qp$, we first stack the $\left\{ \boldhat{A}_{\nf}^\qp \right\}_{\nf=1}^{\Nf}$ into the matrix $\Abar^\qp \triangleq
[(\boldhat{A}_1^\qp)^{\mathsf{T}} \hdots (\boldhat{A}_{\Nf}^\qp)^{\mathsf{T}}]^\mathsf{T} \in \Cbb^{\Nf
  \Nk \times N_z}$. Next, stack $\Nf$ copies of $\mathbf{H}$ into $\Hbar \triangleq
(\ones{\Nf}^{\mathsf{T}} \otimes \mathbf{H}) = [\mathbf{H}^{\mathsf{T}}, \hdots,
\mathbf{H}^{\mathsf{T}}]^\mathsf{T} \in \Cbb^{\Nf \Nk \times \Ns}$. Now, $\Phi^\qp = \Hbar \odot
\mathbf \Abar^\qp$.

\section{The $N$-Species Inverse Problem}
\label{sec:nspecies_inverse}

\subsection{Preliminaries:  The Single Species Case}
\label{sec:isam_properties}
\begin{figure}
  \centering
  \input{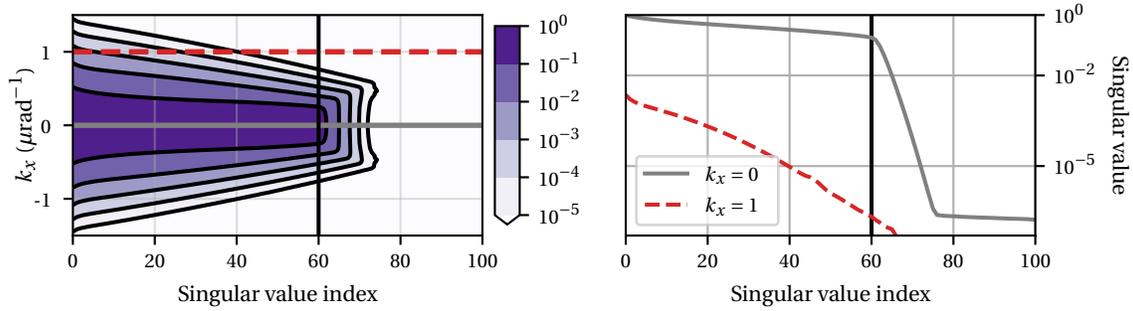}
  \caption{Left: Singular values of $\boldhat{A}^{q_x}_{\nf}$ in the case of one transverse dimension. The coordinate $k_x$ is obtained from
    $q_x$ using \cref{eq:dft_to_ft_coord}. Right: singular values for $k_x=0$ and
    $k_x=1$. The vertical line marks the rank estimate \cref{eq:rank_est}. The focal plane is
    located at $\zf = 140 \si{\micro\meter}$. The remaining system parameters are listed in
    \cref{tab:point_scatter_params}. }
  \label{fig:sv_A_plot} 
\end{figure} 
Under the \nsp model \cref{eq:disc_tensor_form}, the measurements at each focal plane are modeled as the sum of $\Ns$
independent ISAM experiments; thus, the ISAM matrices $\Ahl_{\nf}$ set fundamental limits on
what can be imaged. Stated plainly, if a spatial density lies in the null space of each
$\Ahl_{\nf}$, then it will generate no measurement and thus cannot be imaged using the
proposed method.

A careful study of the spectral properties of these matrices is beyond the scope of this
paper. Instead, we combine a numerical study of these matrices with intuition obtained from the
approximate ISAM kernel \cref{eq:asymp}.
We computed the singular values of $\boldhat{A}^\qp_{\nf}$ in the case of one transverse dimension,
$x$, using the computational parameters listed in \cref{tab:point_scatter_params}. The singular
values are shown in \cref{fig:sv_A_plot}, where the continuous Fourier coordinate $k_x$ is
determined from the DFT index $q_x$ using \cref{eq:dft_to_ft_coord}. While we do not form
$\boldhat{A}^\qp_{\nf}$ using the approximate kernel, the approximate kernel provides intuition for
the behavior seen here. The largest singular values die off quickly as $k_x$ increases, as expected
due to the function $\apfunc{\kp, \ko}$ in \cref{eq:asymp}. Moreover, for $\abs{k_x} > 2 \kmax$,
ISAM matrix is uniformly zero due to $\chi(\kp, 2\ko)$.

According to the approximate forward model \cref{eq:asymp_integral}, for $k_x = 0$ we obtain the
(bandlimited) Fourier transform of the space-limited weighted
susceptibility. The eigenvalue spectrum of space-and-frequency limited Fourier operators has been
studied, beginning with a series of papers by Slepian, Landau, and Pollak
\cite{Slepian1961-Prolat-Spher, Landau1961-Prolat-Spher, Landau1962-Prolat-Spher,
  Slepian1978-Prolat-Spher, Slepian1983}. In the discrete case, the eigenvalue and singular value
spectrum of space-and-frequency limited Discrete Fourier Transform (DFT) matrices have been studied;
such matrices are submatrices formed by consecutive rows and columns of a DFT matrix
\cite{Jain1981-Extrap-Algor, Gruenbaum1981-Eigen-Toepl-Matrix, Zhu2018-Eigen-Distr}. The singular
values of a space-and-frequency limited DFT matrix are divided into three distinct regions: (1) a
region wherein the singular values are near one; (2) a transition region where the singular values
decay exponentially; and (3) the remaining singular values are nearly zero. The number of singular
values in the first region is called the \emph{effective rank} and is written $r_e$. A direct
application of Slepian-Pollak theory predicts \cite{Slepian1978-Prolat-Spher, Zhu2018-Eigen-Distr}
\begin{equation}
  r_e =
  \frac{2 (\kmax - \kmin)}{2 \pi /  L_z} = \frac{L_z}{\pi} (\kmax - \kmin).
  \label{eq:rank_est}
\end{equation}
For fixed $\kp$, the approximate ISAM operator can be viewed as a space-and-frequency limited
Fourier operator with additional weighting in the spatial domain by $\upsilon(z)$ and in the
frequency domain by $\apfunc{\kp, \ko}$. For each $\kp$ the operator is space-limited to a region of
length $L_z$; this is due to assumption that $\eta$ is compactly supported.
Moreover, the operator is frequency-limited to the optical passband $\bset$.
In the discretized setting, only $\mathbf{A}^0_{\nf}$ can be viewed as a (diagonally scaled) DFT
matrix, as for $\qp \neq 0$ the resulting Fourier transform is not uniformly sampled.

We can use the theory of space-and-frequency limited DFT matrices to understand the behavior of the
spectrum of $\boldhat{A}^0_{\nf}$ as shown in \cref{fig:sv_A_plot}. The singular values are broken
into three regions: in the first region, the singular values decay exponentially, albeit at a rate
slower than in the second region. The transition between the first and second regions still occurs at
$r_e$. In the case of the parameters used in \cref{fig:sv_A_plot}, we have $r_e = 60$, and the
change in behavior at $r_e$ is evident. The case of $k_x \neq 0$ is more complicated as the
resulting Fourier transform is not uniformly sampled.

Recall that $\bset$ is the set of observable Fourier components of the weighted susceptibility,
$\nu(z - \zf) \hat{p}(\kp, z)$. A common practice in ISAM imaging is to ignore the axial
weighting function and treat $\bset$ as the observable Fourier components of the \emph{unweighted}
susceptibility (see, \eg \cite{Davis2007, Davis2008}). This is a reasonable approximation of the
imaging system. To justify the approximation, note that $\nu(z)$ is strictly positive and
slowly varying; thus the Fourier transform of the weighted and unweighted susceptibilities are
roughly supported on the same set.

Using the same line of reasoning, we assume that $\nullspace{\boldhat{A}^\qp_{\nf}}$ is invariant to
the choice of focal plane $\zf$. This is reasonable when the focal planes are close to one another.
Note that this is an implicit assumption in previous work on multi-focal ISAM \cite{Xu2014}.

\subsection{Algebraic Conditions for a Unique Solution to \cref{enum:linear}}

We now consider the discretized \nsp inverse problem. We begin by considering the discretized form
of \cref{enum:linear}: we assume the spectral profiles $\mathbf{h}_{\ns}$ are fixed and known. In
this case, recovery of the spatial densities from measurements of the form
\cref{eq:disc_tensor_form} is a linear inverse problem. We study the each ``one-dimensional''
problem $\sbl = \Phil \pbl$ for $\qp \in \left[ N_x \right]^2$, with $\Phi^\qp$ given by
\cref{eq:phi_subblock}. In the remainder of this section the DFT index $\qp$ is fixed.

Without additional constraints on the spatial densities, the existence and uniqueness of a solution
is determined entirely by the matrices $\Phil$. In this section, we establish algebraic conditions
for existence and uniqueness of a solution in terms of the ISAM matrices,
$\left\{\boldhat{A}^\qp_{\nf}\right\}_{{\nf}=1}^{\Nf}$, and the chemical spectra,
$\left\{\mathbf{h}_{\ns}\right\}_{{\ns}=1}^{\Ns}$. Earlier work on this problem claimed that $\Nf
\geq \Ns$ and linear independence of the $\mathbf{h}_{\ns}$ is necessary and sufficient for unique
recovery of the spatial densities within the optical passband\cite{Deutsch2015}. While necessary, we
show these two conditions are not sufficient.

For each focal plane, the ISAM matrix $\Ahl_{\nf}$ is of size $N_k \times N_z$, where $N_k$ is the
number of wavenumber samples and $N_z$ is the (axial) length of the discretized spatial density. Per
\cref{sec:isam_properties}, we assume the null space of $\Ahl_{\nf}$ is invariant to the choice of
focal plane, thus for fixed $\qp$ each matrix has the same rank. Let $r \triangleq \rank{\Ahl_{\nf}}$
for $\nf \in [\Nf]$. We write the shared nullspace of the ISAM matrices as $\nset \subseteq
\Cbb^{N_z}$; we have
\begin{equation}
\nset \triangleq \nullspace{\Ahl_{\nf}} \text{ for } \nf \in [\Nf].
\end{equation}
The optical passband is $\nsetperp$.
Define the subspace 
\begin{equation}
  \nsetbar \triangleq \nset \times \nset \hdots \times \nset =
  \mathrm{span}\Set*{ \ \bar{\mathbf{p}}^\qp
    = [(\boldhat{p}^\qp_{1})^{\mathsf{T}}, \hdots, (\boldhat{p}^\qp_{\Ns})^{\mathsf{T}}]^{\mathsf{T}}
    \given \boldhat{p}^\qp_{\ns} \in \nset,  {\ns} \in [\Ns] } \subseteq \Cbb^{\Ns N_z}
\end{equation}
of block vectors where each block is in $\nset$. The subspace $\nsetbarperp$ consists of block
vectors where each block lies in the optical passband, $\nsetperp$. In an abuse of notation, we
refer to both $\nsetperp$ and $\nsetbarperp$ as ``the optical passband''.

Using the \nsp model, the measurements are a weighted sum of ISAM experiments; thus any objects that
lie in $\nsetbar$ will also be in $\nullspace{\Phil}$. If an object cannot be imaged using ISAM, it
cannot be imaged using $\Phil$. We must consider uniqueness modulo $\nsetbar$; we wish to establish
conditions such that these are the \emph{only} objects that cannot be imaged using $\Phil$. In this
case, the \nsp model does not introduce additional ambiguity and each species is correctly
identified. We do no worse using the \nsp model than if we were able to image the spatial densities
independently using the ISAM system.

Let us pause to consider the geometry of a simple case: two species and a single focal plane. Here,
$\Phil = [\mathbf{D}_1 \Ahl_1, \mathbf{D}_2 \Ahl_1]$ and $\sbl = \Phil
\pbl = \mathbf{D}_1 \Ahl_1 \boldhat{p}_1 + \mathbf{D}_2 \Ahl_1
\boldhat{p}_2$. Clearly, if $\phl_1$ and $\phl_2$ are each in $\nset$, then
$\sbl = \mathbf{0}$. Suppose the $\mathbf{h}_{\ns}$ are non-zero for each index; then
$\mathbf{D}_{\ns}$ is full rank. Using the formula for the rank of a partitioned matrix,
\begin{align}
  \rank{\Phil} = \rank{[\mathbf{D}_1 \Ahl_1, \mathbf{D}_2 \Ahl_1]} 
              &= \rank{\mathbf{D}_1 \Ahl_1} + 
                \rank{\mathbf{D}_2 \Ahl_1}
                - \dimm{\rangespace{\mathbf{D}_1 \Ahl_1} \cap \rangespace{\mathbf{D}_2 \Ahl_1}} \\ 
  &= 2r - \dimm{\rangespace{\mathbf{D}_1 \Ahl_1} \cap \rangespace{\mathbf{D}_2 \Ahl_1}}.
\end{align}
The last term captures the interplay between the $\mathbf{D}_{\ns}$ and $\Ahl_1$. We want to find
conditions under which this intersection is trivial. As we assume $\mathbf{D}_{\ns}$ is full rank, we
can instead ask when $\rangespace{\Ahl_1} \cap \rangespace{\mathbf{D}_1^{-1} \mathbf{D}_2 \Ahl_1}$
is trivial. Loosely speaking, when is multiplication by a diagonal matrix enough to perturb a
subspace out of alignment with itself?

Next, we define a notion of uniqueness modulo the ISAM nullspace.
\begin{definition}
  The solution to $\sbl = \Phil \pbl$
  is said to be \emph{unique within the optical passband} if $\Phil \mathbf{x} = \Phil \mathbf{y}
  \implies \mathbf{x} - \mathbf{y} \in \nsetbar$.
  Equivalently, there is a unique $\pbl \in \nsetbarperp$ such that $\sbl=
  \Phil \pbl$. 
\end{definition} 
This definition sets up an equivalence relation on the spatial densities: we treat two spatial
densities as equivalent if their difference lies in $\nsetbar$, the null space of the ISAM matrices.
This is the component to which we are  inherently are blind even in the single species case.

Next, we cast the problem into a form where we implicitly work in the optical passband
$\nsetbarperp$. Let $\Vl \in \Cbb^{N_z \times r}$ be a basis for $\nsetperp$. We
introduce a new set of matrices: the \emph{restricted ISAM matrix} $\Bhl_{\nf} \triangleq
\Ahl_{\nf} \Vl \in \Cbb^{\Nk \times r}$ is the restriction of $\Ahl_{\nf}$ to
the subspace $\nsetperp$. Clearly, $\Bhl_{\nf}$ has full column rank. Similarly,
$\mathbf{I}_{\Nf} \otimes \Vl$ is a basis for $\nsetbarperp$. We define the
\emph{restricted \nsp matrix}
\begin{equation}
  \Phitl \triangleq \Phil (\mathbf{I}_{\Nf} \otimes \Vl) \in \Cbb^{\Nf \Nk \times \Ns r} .
\end{equation}
The question of unique recovery (within the optical passband) is determined entirely by this matrix,
as stated in the following result.
\begin{lemma}
  Let $\Phil \in \Cbb^{\Nf \Nk \times \Ns N_z}$ and $\rank{\Ahl_{\nf}} = r$ for $\nf \in [\Nf]$.
  The following statements are equivalent:
  \begin{enumerate}[label=(C\arabic*)] 
  \item There is a unique $\pbl \in \nsetbarperp$ such that $\sbl = \Phil
    \pbl$ \label{cond:unique}
  \item $\nullspace{\Phil} = \nsetbar$ \label{cond:nullspace}
  \item $\rank{\Phitl} = \Ns r$. \label{cond:rank}
  \end{enumerate}
  \label{thm:rank_uniqueness}
\end{lemma}
We defer the proof to \cref{sec:proofs}.

We can construct the restricted \nsp matrix $\Phitl$ using the Khatri-Rao product. Let $\Bbarl \in \Cbb^{\Nf
  \Nk \times r}$ be the matrix formed by stacking the restricted ISAM matrices $\Bhl_{\nf}$ into a
single block column: $\Bbarl \triangleq [(\Bhl_1)^{\mathsf{T}}, \hdots,
(\Bhl_{\Nf})^{\mathsf{T}}]^{\mathsf{T}}$. Recall $\Hbar = (\ones{\Nf}^{\mathsf{T}} \otimes \mathbf{H}) =
[\mathbf{H}^{\mathsf{T}}, \hdots, \mathbf{H}^{\mathsf{T}}] \in \Cbb^{\Nf \Nk \times \Ns}$. Now,
\begin{align}
  \Phitl =  \Phil (\mathbf{I}_{\Nf} \otimes \Vl) = 
  \begin{bmatrix}
    \mathbf{\mathbf{D}}_1 \Ahl_{1} \Vl &  \hdots & \mathbf{\mathbf{D}}_{\Ns} \Ahl_{1} \Vl \\
    \vdots       & \ddots & \vdots \\
    \mathbf{D}_1 \Ahl_{\Nf} \Vl &  \hdots & \mathbf{D}_{\Ns} \Ahl_{\Nf} \Vl \\
  \end{bmatrix} =
  \begin{bmatrix}
    \mathbf{\mathbf{D}}_1 \Bhl_1 &  \hdots & \mathbf{\mathbf{D}}_{\Ns} \Bhl_1 \\
    \vdots       &  \ddots & \vdots \\
    \mathbf{D}_1 \Bhl_{\Nf} &  \hdots & \mathbf{D}_{\Ns} \Bhl_{\Nf} \\
  \end{bmatrix}
    = \Hbar \odot \Bbarl,
  \label{eq:phi_tilde_block}
\end{align}
mirroring the construction of $\Phil$ in \cref{sec:khatri-rao}.

In what follows, we establish necessary and sufficient conditions for uniqueness within the optical passband.

\subsubsection{Necessary Conditions for Uniqueness}
\begin{figure}
  \centering
\begingroup%
  \makeatletter%
  \providecommand\color[2][]{%
    \errmessage{(Inkscape) Color is used for the text in Inkscape, but the package 'color.sty' is not loaded}%
    \renewcommand\color[2][]{}%
  }%
  \providecommand\transparent[1]{%
    \errmessage{(Inkscape) Transparency is used (non-zero) for the text in Inkscape, but the package 'transparent.sty' is not loaded}%
    \renewcommand\transparent[1]{}%
  }%
  \providecommand\rotatebox[2]{#2}%
  \newcommand*\fsize{\dimexpr\f@size pt\relax}%
  \newcommand*\lineheight[1]{\fontsize{\fsize}{#1\fsize}\selectfont}%
  \ifx\svgwidth\undefined%
    \setlength{\unitlength}{258.15094234bp}%
    \ifx\svgscale\undefined%
      \relax%
    \else%
      \setlength{\unitlength}{\unitlength * \real{\svgscale}}%
    \fi%
  \else%
    \setlength{\unitlength}{\svgwidth}%
  \fi%
  \global\let\svgwidth\undefined%
  \global\let\svgscale\undefined%
  \makeatother%
  \begin{picture}(1,0.40219492)%
    \lineheight{1}%
    \setlength\tabcolsep{0pt}%
    \put(0,0){\includegraphics[width=\unitlength,page=1]{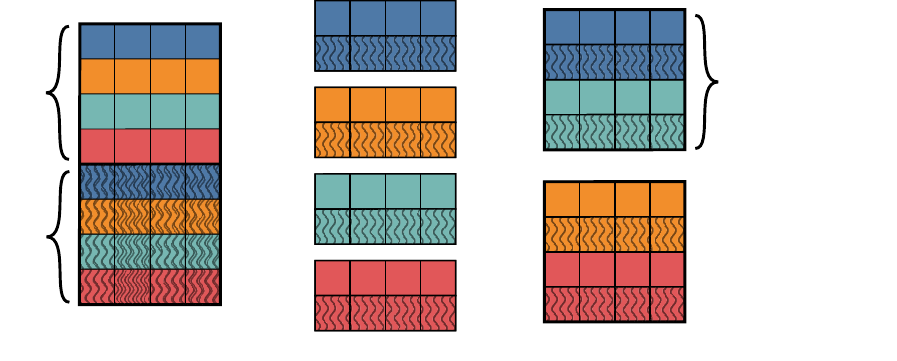}}%
    \put(-0.00192928,0.2935401){\color[rgb]{0,0,0}\makebox(0,0)[lt]{\lineheight{1.25}\smash{\begin{tabular}[t]{l}$\Bhl_1$\end{tabular}}}}%
    \put(-0.00192928,0.13176085){\color[rgb]{0,0,0}\makebox(0,0)[lt]{\lineheight{1.25}\smash{\begin{tabular}[t]{l}$\Bhl_2$\end{tabular}}}}%
    \put(0.80851272,0.29996892){\color[rgb]{0,0,0}\makebox(0,0)[lt]{\lineheight{1.25}\smash{\begin{tabular}[t]{l}$\mathbf{C}_1$\end{tabular}}}}%
    \put(0,0){\includegraphics[width=\unitlength,page=2]{necc_suff_cond_drawing.pdf}}%
    \put(0.80851272,0.1098064){\color[rgb]{0,0,0}\makebox(0,0)[lt]{\lineheight{1.25}\smash{\begin{tabular}[t]{l}$\mathbf{C}_2$\end{tabular}}}}%
    \put(0.15099251,-0.00274641){\color[rgb]{0,0,0}\makebox(0,0)[lt]{\lineheight{1.25}\smash{\begin{tabular}[t]{l}(a)\end{tabular}}}}%
    \put(0.41342357,-0.00274641){\color[rgb]{0,0,0}\makebox(0,0)[lt]{\lineheight{1.25}\smash{\begin{tabular}[t]{l}(b)\end{tabular}}}}%
    \put(0.67018104,-0.00274641){\color[rgb]{0,0,0}\makebox(0,0)[lt]{\lineheight{1.25}\smash{\begin{tabular}[t]{l}(c)\end{tabular}}}}%
  \end{picture}%
\endgroup%

  \caption{Comparing \cref{nec:meas_diversity} and \cref{thm:n_species_suff} for $\Nf=2$ and $r=4$.
    (a) The matrix $\Bbarl$. Color denotes the value of $\ko$. Rows with solid (\emph{resp.}
    wave-patterned) blocks correspond to measurements at the first (\emph{resp.} second) focal
    plane. (b) Condition \cref{nec:meas_diversity} requires that the sum of the ranks of each $2
    \times 4$ block of the same color must be at least $4 \Ns$. (c) A possible partitioning of the
    rows of $\Bbarl$ as described in \cref{thm:n_species_suff}. If both $\mathbf{C}_1$ and
    $\mathbf{C}_2$, as defined in \cref{eq:C_suff_mat}, have full rank for generic chemical species
    the solution to $\sbl = \Phil \pbl$ is unique within the optical passband with probability one.
  }
  \label{fig:necc_suff_cond}
\end{figure}

\begin{theorem}
  The solution to $\sbl = \Phil \pbl$ is unique within the optical passband
  only if
  \begin{enumerate}[label=(N\arabic*)] 
    \item $\Nk \Nf \geq \Ns r$  \label{nec:dimensions}
    \item The spectral profiles are linearly independent ($\rank{\mathbf{H}} = \Ns$)  \label{nec:independent}
    \item No row of $\Bbarl$ is orthogonal to all remaining rows  \label{nec:orthogonality}
    \item For every subset 
      $J \subset [\Nk]$ with $\Ns \leq \abs{J} < \Ns r / \Nf$ and $\rank{\mathbf{H}^J} = \Ns$, we have
      $\rank{\mathbf{H}^{J^c}} \geq \Ns - \frac{\Nf}{r} \abs{J}$  \label{nec:spec_diversity}
    \item 
      $\sum_{i=1}^{\Nk} \rank{\left[\Bhl_1[i, \colon]^{\mathsf{T}}, \hdots, \Bhl_{\Nf}[i,
          \colon]^{\mathsf{T}}\right]^{\mathsf{T}}} \geq \Ns r$. \label{nec:meas_diversity}
  \end{enumerate}
  \label{thm:n_species_necc}
\end{theorem}
  We defer the proof to \cref{sec:proofs}. Let us pause to interpret these conditions.
  
  In the single-species case, \cref{nec:dimensions} reduces to $N_k \geq r$; \ie we must measure
  enough wavenumbers such that the single-species ISAM problem is well-posed. Interestingly,
  \cref{nec:dimensions} does not require that $\Nf \geq \Ns$: recovery of $\Ns$ species is possible
  from a single focal plane, provided the measurements are oversampled in wavenumber. This behavior
  can be seen in the numerical experiments described in \cref{sec:stability}

  Condition \cref{nec:independent} is unsurprising. If the spectral profiles are linearly dependent,
  the \nsp representation of a susceptibility is not unique and the spatial densities cannot be
  uniquely determined.
  
  Condition \cref{nec:orthogonality} is less transparent, but can be argued to hold by the
  underlying physics. If \cref{nec:orthogonality} is violated, there must be an object that scatters
  at only one of the measured wavenumbers and is non-scattering for the rest. In the continuous
  setting, scattered fields are analytic functions of $\ko$; thus if an object is non-scattering
  over an interval of wavenumbers, it must be non-scattering for all $\ko$
  \cite{Wolf1985-Analy-Angul, Devaney1978-Nonun-Inver}.
  In the discretized setting we lose the analytic properties of scattered waves.
  In numerical simulations, however, condition \cref{nec:orthogonality} holds.

  Condition \cref{nec:spec_diversity} requires the spectral profiles to be sufficiently diverse:
  linear independence is not enough. As an example, consider $\Ns = 2, \Nf=1$, and take
  $\mathbf{h}_1 = [1, 1, \hdots, 1]^{\mathsf{T}}$ and $\mathbf{h}_2 = [2, 1, \hdots,
  1]^{\mathsf{T}}$. These spectra are linearly independent, but $\mathbf{D}_1 \Ahl_{1}$ and
  $\mathbf{D}_2 \Ahl_1$ differ by only one row; thus $\rank{\Phitl} \leq r + 1$, failing
  \cref{cond:rank} of \cref{thm:rank_uniqueness}. Spectral diversity is necessary to push
  $\rangespace{\mathbf{D}_1 \Ahl_1}$ out of alignment with $\rangespace{\mathbf{D}_2 \Ahl_1}$.
  ``Good'' spectral profiles are not too concentrated on any small set of indices.

  
  The final condition, \cref{nec:meas_diversity}, is a requirement on the diversity of measurements
  comprising the restricted ISAM matrices. When $\Nk \Nf = \Ns r$, \cref{nec:meas_diversity} requires
  that the collection of measurement vectors corresponding to a given wavenumber be linearly
  independent: each new focal plane must provide new and informative measurements. This partitioning
  is illustrated in \cref{fig:necc_suff_cond}.

\subsubsection{Sufficient Condition for Uniqueness}
\label{sec:suff}
First, we note that no conditions on $\Bbarl$ or $\mathbf{H}$ independently are sufficient
to ensure there is a unique solution within the optical passband. Consider again
the two-species, one focal plane case: $\Phitl =
[\mathbf{D}_1 \Bhl_1, \mathbf{D}_2 \Bhl_1]$, with $\mathbf{D}_i = \diag{\left( \mathbf{h}_i
  \right)}$. Suppose $\mathbf{h}_1$ is fixed and choose vectors $\mathbf{w}, \mathbf{v} \in \Cbb^r$
such that no entry of $\Bhl_1 \mathbf{v}$ is zero. Set $\mathbf{h}_2 = (\mathbf{D}_1 \Bhl_1 \mathbf{w})
/ (\Bhl_1 \mathbf{v})$ where the division is taken elementwise. With this construction, $\mathbf{D}_2
\Bhl_1 \mathbf{v} = \mathbf{D}_1 \Bhl_1 \mathbf{w}$, and thus $\rank{\Phitl} \leq 2 r - 1$, failing
\cref{cond:rank} of \cref{thm:rank_uniqueness}.

These spectral profiles were carefully chosen to make $\Phitl$ lose rank.
Fortunately, we are unlikely to encounter such objects in practice.  The following definition makes
this argument precise.
\begin{definition}
  A property that depends on the spectral profiles $\mathbf{H} \in \Cbb^{\Nk \times \Ns}$ is said to
  hold \emph{generically}, or \emph{for generic $\mathbf{H}$}, 
  if the set for which it fails to hold has Lebesgue measure zero and is
  nowhere dense in $\Cbb^{\Nk \times \Ns}$.
\end{definition}
If a property that holds generically, it holds with probability one if the spectral profiles are
drawn independently from a distribution that is absolutely continuous with respect to the Lebesgue
measure in $\Cbb^{\Nk \times \Ns}$; for instance, when the entries of $\mathbf{H}$ are drawn i.i.d.
from the Gaussian distribution. Moreover, the property exhibits a degree of robustness: if it holds
for a particular $\mathbf{H}^\prime$, then it holds in an open ball around $\mathbf{H}^\prime$ and
will continue to hold given sufficiently small perturbations to $\mathbf{H}^\prime$.

\begin{theorem}
  Suppose $\Nk \geq r$ and $\Nf \geq \Ns$.  If there exists a
  collection $\left\{ J_i \subset [\Nk] \right\}_{i=1}^{\Nf}$ of disjoint sets, each of cardinality
  $\abs{J_i} = r / \Nf$, such that
  \begin{equation}
    \mathbf{C}_i \triangleq \begin{bmatrix}
      \Bhl_1[J_i, \colon] \\
      \vdots \\
      \Bhl_{\Nf}[J_i, \colon]
    \end{bmatrix}
    \in \Cbb^{r \times r}
    \label{eq:C_suff_mat}
  \end{equation}
  is full rank for each $i \in [\Nf]$, then
  for generic $\mathbf{H}$ the solution to $\sbl = \Phil \pbl$ is unique
  within the optical passband.
  \label{thm:n_species_suff}
\end{theorem}
An illustration of the matrices $\mathbf{C}_i$ is shown in \cref{fig:necc_suff_cond}(c). Note that
the necessary condition \cref{nec:meas_diversity} coincides with the sufficient condition of
\cref{thm:n_species_suff} in the case of $\Nk = \Nf = r = \Ns$, which is the limit of scanning
confocal spectroscopic acquisition discussed in \cref{sec:intro}.

\cref{thm:n_species_suff} can be stated in terms of a more
familiar, but more restrictive, property on $\Bbarl$.
\begin{definition}
  The \emph{Kruskal (row) rank} of a matrix $\mathbf{X} \in \Cbb^{n \times m}$, written
  $\krank{\mathbf{X}}$,
  is the largest $k$ such that every
  set of $k$ rows of $\mathbf{X}$ are linearly independent.  The matrix $\mathbf{X}$ is said to have
  \emph{full Kruskal rank} if $\krank{\mathbf{X}} = \max\left\{n, m\right\}$. 
\end{definition}
\begin{corollary}
  If $\Bbarl \in \Cbb^{\Nk \Nf \times r}$ has full Kruskal rank, then
  for generic $\mathbf{H}$ the solution to $\sbl = \Phil \pbl$ is unique within the optical passband.
\end{corollary}

\subsubsection{Related Problems}
\label{sec:related_kr}
The Khatri-Rao structure of $\Phi^\qp$ provides a link between the \nsp inverse problem and topics in
tensor factorization, communications, and sensor networks, among others
\cite{Sidiropoulos2000-Uniquen-Multil, Khanna2019-Correc-To, Khanna2018-Restr-Isomet,
  Fengler2019-Restr-Isomet, Allman2009-Ident-Param, Bhaskara2014-Smoot, Gamal2015-When}. For
example, the rank and Kruskal rank of the Khatri-Rao product has implications for the uniqueness of
certain tensor factorizations. Properties of the Khatri-Rao product are an active area of research.
For generic matrices $\mathbf{X}$ and $\mathbf{Y}$, it is known that $\krank{\mathbf{X} \odot
  \mathbf{Y}} = \krank{\mathbf{X}} \krank{\mathbf{Y}}.$ Bhaskara \etal provide bounds on the
smallest singular value of the Khatri-Rao product of random matrices \cite{Bhaskara2014-Smoot}.
Recent work has investigated the restricted isometry property of the Khatri-Rao product of random
matrices \cite{Khanna2019-Correc-To, Khanna2018-Restr-Isomet, Fengler2019-Restr-Isomet}.

These results do not directly apply to the problem at hand. We are interested in properties
of $\Phitl = \Hbar \odot \Bbarl$. As $\Bbarl$ is determined by the physics and imaging geometry, we
cannot choose this matrix generically or randomly. Even $\mathbf{\bar{H}}$ cannot be chosen
generically, as $\mathbf{\bar{H}} = (\ones{\Nf}^{\mathsf{T}} \otimes \mathbf{H})$; only the matrix
$\mathbf{H}$ can be chosen generically. Translating new results on the Khatri-Rao product to our
setting remains a topic for further investigation.

\subsection{Stability And Conditioning of \cref{enum:linear}}
\label{sec:stability}
The results of the previous section tell us that the solution to $\mathbf{\bar{s}}^\qp = \Phi^\qp \pbl$
is almost always unique (within the optical passband), but say little about the stability of the
problem. We must always deal with ``noisy'' measurements-- not just instrumentation noise, but also
``noise'' due to modeling error, \eg multiple scattering and spatial-spectral coupling not captured
by the \nsp model.


In this section, we numerically investigate the behavior of the singular values of the \nsp matrix
$\Phi$ for the case three-species case $(\Ns = 3)$ in two spatial dimensions. We use the
computational parameters listed in \cref{tab:point_scatter_params}, except for $\na$ and $\Nf$,
which vary. The singular values of the ISAM matrix formed using these computational parameters were
investigated in \cref{sec:isam_properties} and plotted in \cref{fig:sv_A_plot}. The spectral
profiles used---caffeine, acetaminophen, and warfarin---are shown in \cref{fig:points_spectra}.

We computed the singular values of each block-matrix $\Phi^\qp$ \cref{eq:phi_subblock} and plot the
results in \cref{fig:sv_vs_na_zf_fig}. Recall
that the continuous Fourier frequency $k_x$ is determined from the DFT index $q_x$
using \cref{eq:dft_to_ft_coord}. As expected, higher transverse spatial frequencies
are present as $\na$ increases. Only the first $\Nf r_e$ singular values are appreciable. The
low-frequency components achieve rank $3 r_e$ for $\Nf = 3$, and adding focal planes improves the
conditioning of $\Phi$. Note that even in the case of a single focal plane, the $3 r_e$-th singular
value of $\Phi^0$ is non-zero; as previously discussed, $\Nf \geq \Ns$ is not necessary for a unique
solution.

We investigated the singular values of the block corresponding to $k_x=0$ for a variety of chemical
species and a varying number of focal planes. We used a library of $20$ experimentally acquired
chemical spectra\footnote{These include caffeine, acetaminophen, warfarin, monosodium glutamate (MSG), sucrose,
  naproxen, potassium chlorate, polyvinylidene fluoride (PVDF), aspartame, lactose, melatonin,
  ethylenediaminetetraacetic acid (EDTA), creatine, diazepam, biotin, fructose, pectin, glycine,
  beta carotene, hydroxypropyl cellulose.} provided through the IARPA SILMARILS project. We randomly
selected three species from the library, formed $\Phi^{0}$, and computed the singular values of
this matrix. We scaled $\Phi^{0}$ to have unit spectral norm. This procedure was repeated for
$200$ realizations. The resulting singular values are plotted in \cref{fig:sv_ensemble}; the borders
of the shaded region are the best and worst realizations for each choice of $\Nf$.

We repeated the same procedure using random spectral profiles. The real part of the spectral profile
was drawn i.i.d. from the standard normal distribution and the imaginary part was chosen uniformly
over $[0, 1]$. The results are plotted in \cref{fig:sv_ensemble}. Clearly, these un-physical spectra
lead to better conditioned $\Phi^0$, and there is little difference in the best and worst
realizations. Study of the system using random spectral profiles may lead to a useful upper bound on
system performance.

 \begin{figure}[t]
   \centering
   \input{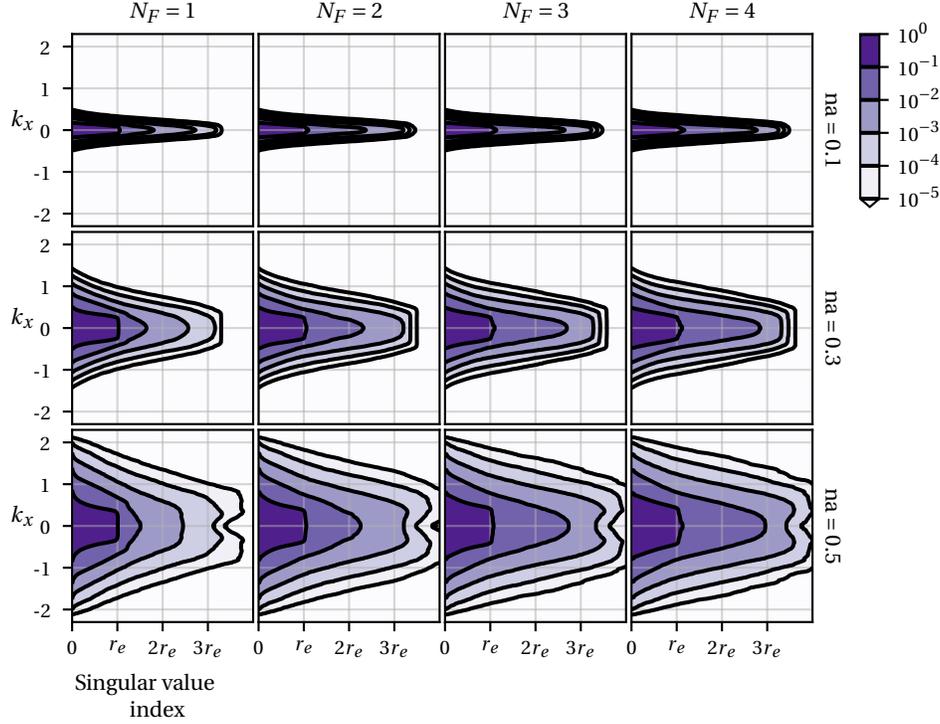}
   \caption{Singular values of $\Phi^{q}$ as a function of $k_x(q)$. Three
     species are present: caffeine, acetaminophen, and warfarin. System parameters listed in
     \cref{tab:point_scatter_params}. }
   \label{fig:sv_vs_na_zf_fig}
 \end{figure}
 \begin{figure}[h]
   \centering
   \input{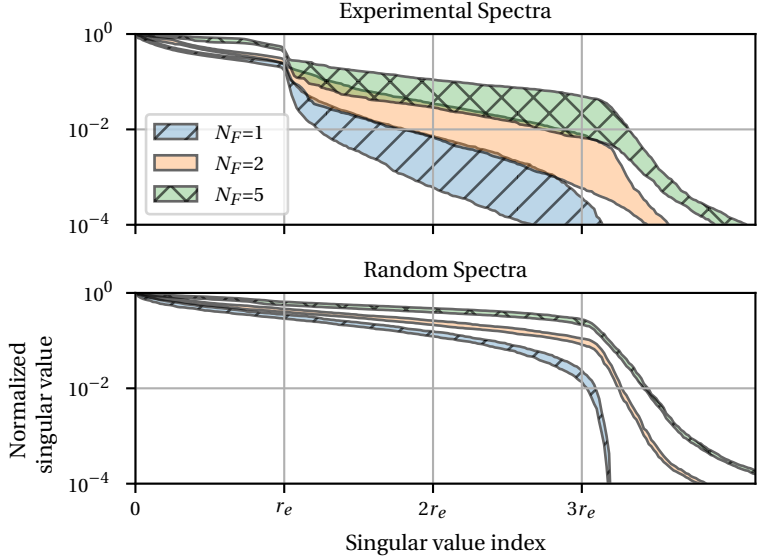} 
   \caption{Singular values of $\Phi^0$ for various combinations of chemical species. The shaded
     area lies in between the best and worst realizations. System parameters listed in
     \cref{tab:point_scatter_params}. Top: singular values using experimentally acquired spectral
     profiles. Bottom: singular values using random Gaussian spectral profiles. }
   \label{fig:sv_ensemble} 
 \end{figure} 

\subsection{Algebraic Conditions for \cref{enum:dictionary}}
We now focus on the case \cref{enum:dictionary}, wherein the target comprises $\Ns$ chemical species
drawn from a ``dictionary'' of $M_s > \Ns$ possible spectra. This problem can be viewed as an instance of
\cref{enum:linear}, in which case \cref{thm:n_species_necc} requires that number of focal planes is
chosen such that $\Nf \Nk \geq M_s r$. This is undesirable if $M_s$ is much larger than $\Ns$.
This approach ignores the constraint that only $\Ns$ chemicals are present in
the sample; by incorporating this side information, we relax the condition on $\Nf$.  This
structure is known as \emph{block sparsity}.
\begin{definition}
  The block vector
  $ \mathbf{\overline{x}} = [\mathbf{x}^{\mathsf{T}}_1,  \mathbf{x}^{\mathsf{T}}_2, \hdots, 
  \mathbf{x}^{\mathsf{T}}_{M_s}]^{\mathsf{T}}$ is said to be \emph{block-K sparse} if the set
  $\Set*{i : \norm{\mathbf{x}_i}_2 > 0}$ has cardinality at most $K$.
\end{definition}
Block sparsity is a natural fit for this problem; we define the $\ns$-th block to be the $\ns$-th
spatial density $\boldhat{p}_{\ns}$, corresponding to the $\ns$-th species in the dictionary. Note
that block sparsity does not require the blocks themselves (\ie, the $
\left\{\boldhat{p}_{\ns}\right\}_{\ns=1}^{\Ns}$) to be sparse.

Conditions for unique recovery of block-sparse vectors have been studied \cite{Eldar2009, Eldar2010,
  Eldar2009c, Elhamifar2012}. Eldar and Mishali \cite{Eldar2009c} developed a straightforward condition for unique
recovery that is useful to the problem at hand:
\begin{lemma}{\cite[Proposition 1]{Eldar2009c}}
  There is a unique block-$\Ns$ sparse solution to $\sbl
  = \Phil \pbl$ if and only if $\Phil \mathbf{v} \neq 0$ for any non-zero $\mathbf{v}$ that is
  block-$2 \Ns$ sparse.
  \label{thm:block_krank}
\end{lemma}
We can easily translate \cref{thm:block_krank} into this setting.
\begin{theorem}
  For generic $\mathbf{H}$, within the optical passband there is a unique block-$\Ns$ sparse vector
  $\pbl$ consistent with measurements $\sbl = \Phil \pbl$ if $N_k > r$, $\Nf \geq 2 \Ns$, and
  $\Bbarl$ contains $2 \Ns$ disjoint sets of linearly independent rows, each of cardinality $r =
  \rank{\Bbarl}$.
\end{theorem}
\begin{proof}
  Let $\mathbf{v}$ be a block-$2 \Ns$ sparse vector. Let $J = [j_1, \hdots, j_{2
    \Ns}]^{\mathsf{T}} \in \Zbb^{2 \Ns}$ index the non-zero blocks of $\mathbf{v}$. The vector
  $\mathbf{v}_{J} \in \Cbb^{2 \Ns N_z}$ contains the non-zero elements of $\mathbf{v}$. The
  matrix $\Phil_J \in \Cbb^{\Nf \Nk \times 2 \Ns N_z}$ is the restriction of $\Phil$ to the $2
  \Ns$ columns indexed by $J$.

  By assumption, $\Bbarl$ satisfies the conditions of \cref{thm:n_species_suff} and $\Phil_J$ is
  generically full column rank. Thus, for generic $\mathbf{H}$, we have $\Phil \mathbf{v} =
  \Phil_J \mathbf{v}_J \neq 0$. Applying \cref{thm:block_krank} completes the proof.
\end{proof}

\subsection{Computational Recovery} 
In the single-species case, the approximate form of the ISAM operator (\cref{sec:asymptotics})
provides a non-iterative reconstruction based on Fourier resampling \cite{Ralston2006a}.
This does not carry over to the multi-species case and we must instead use an iterative approach.

To ease notation, define $\boldhat{p}_s \triangleq \vect{\boldhat{P}_s}$.
We recover the collection of spatial densities by solving the penalized least squares problem
\begin{equation}
  \argmin_{\left\{\boldhat{p}_1, \hdots, \boldhat{p}_{\Ns} \right\}}  \frac{1}{2}
  \sum_{\qp} \norm{\mathbf{\bar{s}}^{\qp} - \Phil \mathbf{\bar{p}}^{\qp}}_2^2 +
  R\left(\boldhat{p}_1, \hdots, \boldhat{p}_{\Ns}\right).
  \label{eq:pwls}
\end{equation}
The first term is known as the \emph{data fidelity} term. It ensures the observed data 
and ``re-imaged'' solution  are consistent. More sophisticated data fidelity terms
can be used to model the effects of shot noise, background signal, and more \cite{Erdogan1999}, but
these are beyond the scope of this work.

The functional $R : \Cbb^{\Ns \times N_x \times N_x \times N_z} \to \Rbb$ regularizes the inverse
problem and encodes any constraints or \emph{a priori} assumptions regarding the spatial densities.
Tikhonov regularization corresponds to
$R\left(\boldhat{p}_1, \hdots, \boldhat{p}_{\Ns}\right) = \sum_{\ns=1}^{\Ns} \norm{\boldhat{p}_s}_2^2$.
Alternatively, solutions that are sparse in a transform domain are
obtained by setting
$R\left(\boldhat{p}_1, \hdots, \boldhat{p}_{\Ns}\right) = \sum_{\ns=1}^{\Ns} \norm{\mathbf{C} \boldhat{p}_s}_1$.
where $\mathbf{C}$ is a sparsifying transform, \eg a wavelet transform. Finally, the mixed
$\ell_1/\ell_2$ norm $\sum_{\ns=1}^{\Ns} \norm{\boldhat{p}_s}_2$ encourages solutions that are
block-sparse; that is, solutions with a minimal number of active species. The non-negative scalar
$\lambda_r$ balances the influence of the data fidelity and regularization terms.

The method used to solve \cref{eq:pwls} depends on the chosen regularizer. In the case of Tikhonov
regularization, \cref{eq:pwls} reduces to the solution of the linear system
\begin{equation}
  \left(\left( \Phil \right)^{\mathsf{H}} \Phil + \lambda_r \mathbf{I}\right) \mathbf{\bar{p}}^{\qp}  = \left(\Phil \right)^{\mathsf{H}} \mathbf{\bar{s}}^{\qp}
\end{equation}
for each $\qp \in \left[ N_x \right]^2$.
The conjugate gradient algorithm works well in practice and requires only matrix-vector products
with $\Phil$ and $\left( \Phil \right)^{\mathsf{H}}$.  These matrices are not explicitly formed;
only the coefficients $\Ahl_{\nf}[\qp, m, n]$ in \cref{eq:isam_matrix_coeff} are precomputed and stored.  Similarly, the
matrices $\mathbf{D}_s$ are not formed;  only the spectral profiles are stored, and products with
$\mathbf{D}_s$ are computed by elementwise multiplication.  The vector $\mathbf{\bar{y}}^\qp = \Phil
\mathbf{\bar{p}}^\qp$ consists of blocks $\boldhat{y}^{\qp}_{\nf} \in \Cbb^{\Nk}$ for $\nf \in
[\Nf]$, with
\begin{equation}
  \mathbf{\hat{y}}^\qp_{\nf} = \sum_{\ns=1}^{\Ns} \mathbf{D}_{\ns} \Ahl_{\nf} \mathbf{\hat{p}}^\qp_{\ns}.
\end{equation}
Similarly, $\mathbf{\bar{w}}^{\qp} = \left(  \Phil\right)^{\mathsf{H}}
\bar{\mathbf{y}}^{\qp}$ consists of blocks $\mathbf{\hat{w}}_{\ns}^\qp$ with $\ns \in \Ns$, where the block is
computed as
\begin{equation}
\mathbf{\hat{w}}^\qp_{\ns} = \sum_{\nf=1}^{\Nf} (\Ahl_{\nf})^{\mathsf{H}}\mathbf{D}^{\mathsf{H}}_{\ns} \mathbf{\hat{y}}^\qp_{\nf}.
\end{equation}
Many sparsity-promoting regularizers are non-differentiable. In this case, proximal methods such as
FISTA \cite{Beck2009} or the Alternating Direction Method of Multipliers (ADMM) \cite{Eckstein1992,
  Boyd2010, Afonso2011} are attractive. This class of algorithms decomposes the problem
\cref{eq:pwls} into a sequence of simpler subproblems. The solution of a linear system 
is often a key ingredient of such algorithms.

\section{Simulations}
\label{sec:nspecies_sims}
We now describe two simulations used to validate the proposed approach. For simplicity, we consider
only two spatial dimensions: one transverse ($x$) and one axial ($z$).

Preliminary work on the \nsp model suffers from three unrealistic assumptions \cite{Deutsch2015}.
The simulations used unrealistic wavelength ranges, leading to nearly complete coverage of Fourier
space. This removes the large null space present in $\mathbf{A}_{\nf}$ and simplifies the
reconstruction problem. Secondly, the phantoms used satisfied the \nsp model exactly; no spectral
noise was considered. Finally, the synthetic data used in the simulations were generated using
the asymptotic approximation to the ISAM operator, and thus
under the first Born approximation. This neglects multiple scattering, absorption, and the
discrepancy between the exact and approximate ISAM models. As a consequence, the simulations present
an overly optimistic view of the proposed imaging modality.

We generate synthetic data using accurate physical models and system parameters. These data
include multiple scattering and absorption effects---only the inversion is performed under the Born
approximation. Further, the simulated targets do not precisely follow the \nsp model; instead, there
are position-dependent spectral variations within each species. In particular, we simulate an object
of the form $\eta(\rr, \ko) = \sum_{\ns=1}^{\Ns} p_{\ns}(\rr) h_{\ns}(\rr, \ko)$, where $h_{\ns}(\rr, \ko) =
h_{\ns}(\ko) + e_{\ns}(\rr, \ko)$ and $e_{\ns}(\rr, \ko) \sim \mathcal{CN}(0, \xi_{\ns})$ is a circular complex
Gaussian random variable \cite{Goodman2015-Statis-Optic}.

The minimization problem \cref{eq:pwls} is solved on an NVidia Titan X GPU using a combination of
Python and CUDA \cite{scikitscuda, Kloeckner2012-Pycud-Pyopen}.

\subsection{Point Targets}
\label{sec:point_targets}
\begin{table}
  \centering
  \begin{tabular}{|c  c || c  c || c  c |}
    \hline
   $N_x$ & $192$ & $L_x$ & $423.6 \si{\micro\meter}$   & $\Delta_x$ & $2.2 \si{\micro\meter}$ \\ \hline
    $N_z$ & $384$ & $L_z$ & $282.4 \si{\micro\meter}$   & $\Delta_z$ & $0.7 \si{\micro\meter}$ \\ \hline
    $\Nk$ & $384$ & $\kmin$ & $0.4 \ \si{\radian \cdot \micro\meter}^{-1}$ & $\kmax$ & $1.1\ \si{\radian \cdot \micro\meter}^{-1}$\\ \hline 
    $r_e$ & $60$& $\lmin$ & $5.9 \si{\micro\meter}^{-1}$ & $\lmax$    & $15.4 \si{\micro\meter}^{-1}$ \\ \hline 
     $\Nf$ & $3$   & $\zf$ & $[70, 140, 211] \si{\micro\meter}$ & $\na$   & $0.4$ \\ \hline
  \end{tabular}
  \caption{Parameters for point target simulations.}
  \label{tab:point_scatter_params}
\end{table}
We formed a spectral library of five chemicals using refractive index data provided through the
IARPA SILMARILS project. The corresponding spectral profiles
are plotted in \cref{fig:points_spectra}. The target consisted of $50$ point scatterers. Each point
scatterer is associated to one chemical species; only three species (out of the five possible) are
present. We do not know \emph{a priori} which chemicals are present.

We generated measurements using the Foldy-Lax model, which includes multiple scattering effects
\cite{Devaney2012}. Data were generated at three focal planes in a $420 \times 280 \ \si{\micro\meter}$
volume according to the parameters in \cref{tab:point_scatter_params}. The source power spectrum was
flat over $[\kmin, \kmax]$.  This combination of parameters---three active species, three focal
planes, and a library of five possible species---corresponds to the case of
\cref{{enum:dictionary}}.

To assess the deviation from the single scattering model, we generated two sets of measurements
using the same target. The first set of measurements, denoted $\mathbf{s}$, uses the Foldy-Lax
method and incorporates multiple scattering. The second, $\mathbf{s}_B$, is generated using the Born
approximation and thus includes only single scattering events. The ratio $\norm{\mathbf{s} -
  \mathbf{s}_B}_2 / \norm{\mathbf{s}_B}_2$ indicates that more than $20\%$ of the energy in
$\mathbf{s}$ comes from multiple scattering events.

We performed two sets of simulations: the first using Tikhonov regularization and the second using
sparsity-promoting regularization. In the latter case, motivated by the spatial-domain sparsity of
the target, we set $R(\mathbf{P}) = \sum_{\ns=1}^{5} \norm{\mathbf{p}_{\ns}}_1$. In the Tikhonov case,
we performed $300$ iterations of conjugate gradient on the normal equations with $\lambda_r=10^{-5}$.
In the case of $\ell_1$ regularization, we used $2000$ iterations of the FISTA algorithm with
$\lambda_r =10^{-3}$. Both cases terminated in under one minute.

The magnitude of the reconstructed spatial densities are shown in \cref{fig:points_recon_tiled}.
Recall that the surface of observable Fourier components is restricted to $\kzo < 0$. As such, any
linear reconstruction method (\eg, Tikhonov-regularized least squares) will produce a complex-valued
image; we display only the magnitude and squared magnitude of the recovered signal. For
visualization purposes we have projected the point-target phantom onto the optical passband. In both
cases, the reconstructed targets are correctly spatially localized and identified with the correct
species.

The Tikhonov regularized reconstruction consists of the point scatterers sitting on top of a
``noisy'' background. The background is primarily due to multiple scattering effects and spectral
variations which are not captured by the forward model. This background term is distributed across
all five possible species; however, the recovered point scatterers are associated to the correct
species.  The background is eliminated when viewing the squared modulus of the
reconstruction.  

The $\ell_1$ regularized reconstruction suppresses the background term. There is nearly perfect
agreement between the true target and the reconstructed target, despite taking data at only three,
rather than five, focal planes. The sparsity of the target, coupled with the $\ell_1$ regularization,
successfully eliminates artifacts due to multiple scattering.

For visualization purposes we map the three active species to the red, green, and blue channels of
an RGB image. The filtered phantom, Tikhonov, and filtered $\ell_1$ reconstructions
are shown in \cref{fig:points_rgb}. 

\begin{figure}
  \centering
  \input{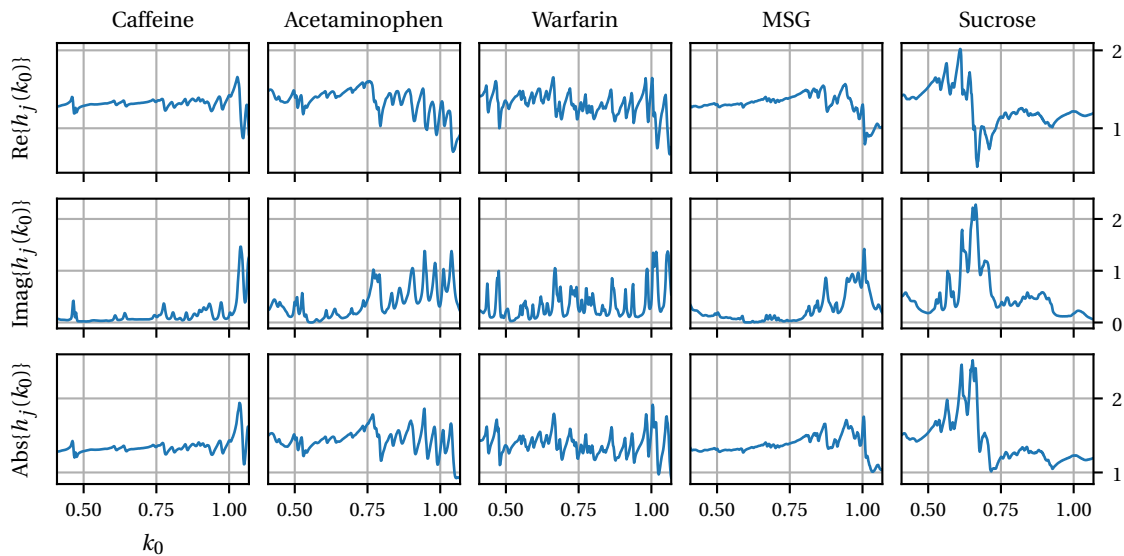}
  \caption{Spectral profiles for the five chemicals used in point scattering simulations.}
  \label{fig:points_spectra}
\end{figure}

\begin{figure}
  \centering
  \input{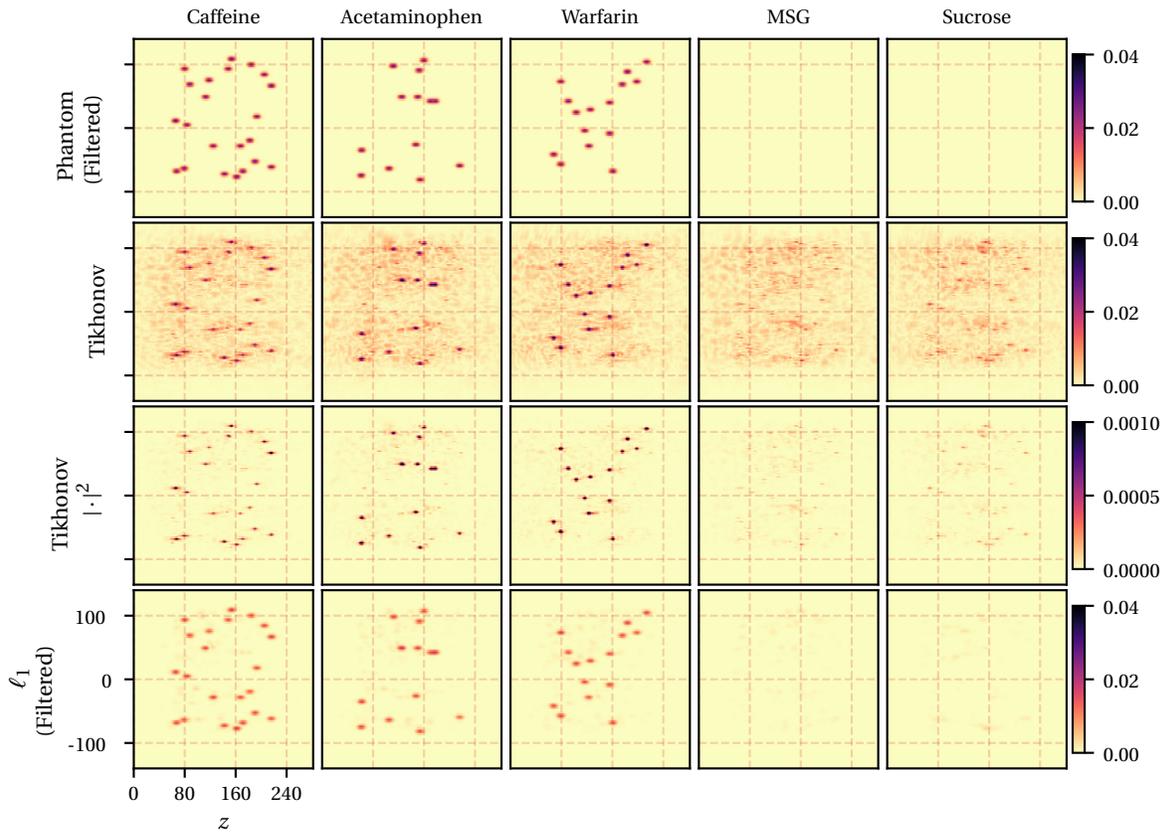}
  \caption{Reconstructions of point scatterers described in \cref{sec:point_targets}.  }
  \label{fig:points_recon_tiled}
\end{figure}

\begin{figure}
  \centering
  \input{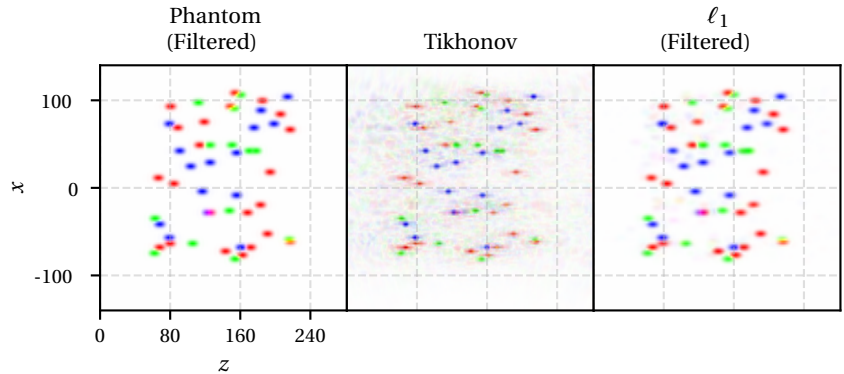}
  \caption{ Visualizing the reconstructed point targets by assigning the three active species to a
    RGB channel. Red: caffeine, Green: acetaminophen, Blue: warfarin. The two species that are not
    present (MSG, Sucrose) were ignored. }
  \label{fig:points_rgb}
\end{figure}

\subsection{Cell Phantom}
\label{sec:cell_phantom}
\begin{table}[ht]
  \centering
  \begin{tabular}{|c  c || c  c || c  c |}
    \hline
$N_x$ & $256$ & $L_x$ & $150.0 \si{\micro\meter}$   & $\Delta_x$ & $0.6 \si{\micro\meter}$ \\ \hline
    $N_z$ & $256$ & $L_z$ & $150.0 \si{\micro\meter}$   & $\Delta_z$ & $0.6 \si{\micro\meter}$ \\ \hline
    $\Nk$ & $256$ & $\kmin$ & $0.7 \ \si{\radian \cdot \micro\meter}^{-1}$ & $\kmax$ & $2.1\ \si{\radian \cdot \micro\meter}^{-1}$\\ \hline 
    $r_e$ & $67$& $\lmin$ & $3.0 \si{\micro\meter}^{-1}$ & $\lmax$    & $9.0 \si{\micro\meter}^{-1}$ \\ \hline 
     $\Nf$ & $3$   & $\zf$ & $[54, 75, 96] \si{\micro\meter}$ & $\na$   & $0.5$ \\ \hline
  \end{tabular}
  \caption{Parameters for cell phantom simulation.}
  \label{tab:cell_params}
\end{table}
Next, we evaluated the ability to image extended targets. The target is the cellular phantom shown
in \cref{fig:annotated_phantom}, which comprises three chemical species. The spectral library
contains five total species.

We generated synthetic measurements by solving the Lipmann-Schwinger equation (see,
\eg,\cite{Devaney2012}) using the using the Multi-Level Fast Multipole Algorithm (MLFMA)
\cite{Coifman1993-Fast-Multip}.  The data are not generated under the Born approximation, and thus
includes multiple scattering and absorption phenomenon not captured using the forward model.
We use a version of the MLFMA specialized for simulating two spatial
dimensions \cite{Meng2018-Wideb-Fast, Hidayetoglu2018-Fast-Massiv}.

We generated measurements for only three focal planes; the relevant computational parameters are
listed in \cref{tab:cell_params}.
We generated synthetic spectral profiles using a sum-of-Lorentzians model
\cite{Fowles1989-Introd}.  Each spectral profile is of the form
 \begin{equation}
   \label{eq:lorentzian_species}
   h(\ko) \triangleq \sigma_{0} +  \sum_{n=1}^{99}  \frac{\sigma_{n}}{\nu^2_{n} - \ko^2 - \mathrm{i} \gamma_{n} \ko},
 \end{equation}
 with $\sigma \sim \mathrm{Unif}[0, 0.1]$, $\nu \sim \mathrm{Unif}[1.2 \pi, 4.4 \pi]$,
 and 
 $\gamma \sim \mathrm{Unif}[\num{2 \pi e-3}, \num{4\pi e-2}]$, where $\mathrm{Unif}[a, b]$ is
 the uniform distribution over the interval $[a, b]$. The spectral profiles are plotted in
 \cref{fig:cell_spectra}.

The first-order Born approximation is valid only if the total phase change between the incident field
and the field inside the sample is less than $\pi$---this implies that the object should be either
weakly scattering or small in spatial extent \cite{Chen1998,Slaney1984}. The proposed phantom is
neither. To investigate the effect on scattering strength on the reconstructed images, we generated
synthetic measurements for the scaled object $\delta \eta(\ro, \ko)$ where $0 < \delta \leq 1$. By
reducing $\delta$, we reduce the scattering strength and eventually fall into a regime where the
first-order Born approximation holds.

We used Tikhonov regularization with $\lambda_r = \num{1e-4}$ and $500$ iterations of the
conjugate-gradient algorithm. The resulting reconstructions are shown in
\cref{fig:cell_by_species_noisy}. The top row illustrates the projection of the phantom onto the
optical passband; this serves as the ``gold standard'' for Tikhonov-regularized
reconstructions. The remaining rows are the reconstructed images. As expected, only the edges of the
phantom that are nearly perpendicular to the optical axis are visible. The reconstructed images
deteriorate as $\delta$ increases, particularly at the rear edge of each feature. However, the
correct species is identified in each case; negligible energy is deposited into Species 4 and 5.

\cref{fig:cell_tik} illustrates the influence of the regularization parameter $\lambda_r$. Noise
dominates the reconstruction when $\lambda_r$ is too small. When $\lambda_r$ is too large, there is
no chemical identification- the recovered spatial densities are nearly identical for each species.

\begin{figure}
  \centering
\begin{subfigure}[t]{2in} 
  \input{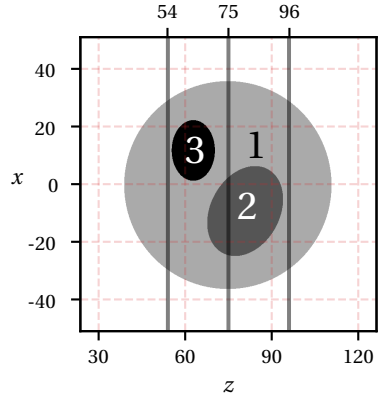}
  \caption{}
  \label{fig:annotated_phantom} 
\end{subfigure}
\\
\begin{subfigure}[t]{\textwidth} 
  \centering
  \input{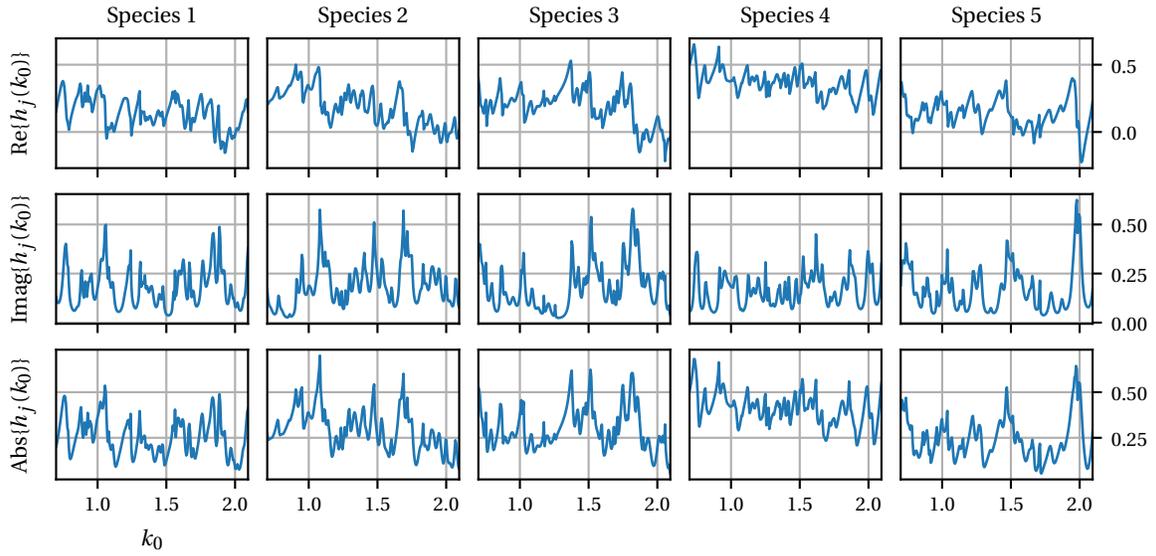}
  \caption{}
  \label{fig:cell_spectra}
\end{subfigure}
\caption{(a)  Three-species cell phantom. The detector plane is plane located at $z=0$. Solid
  vertical lines denote the three focal planes. All units are $\si{\micro\meter}$.
(b)Spectral profiles for cell phantom, plotted for $\delta = 1$.}
\end{figure}

\begin{figure}
  \centering
  \input{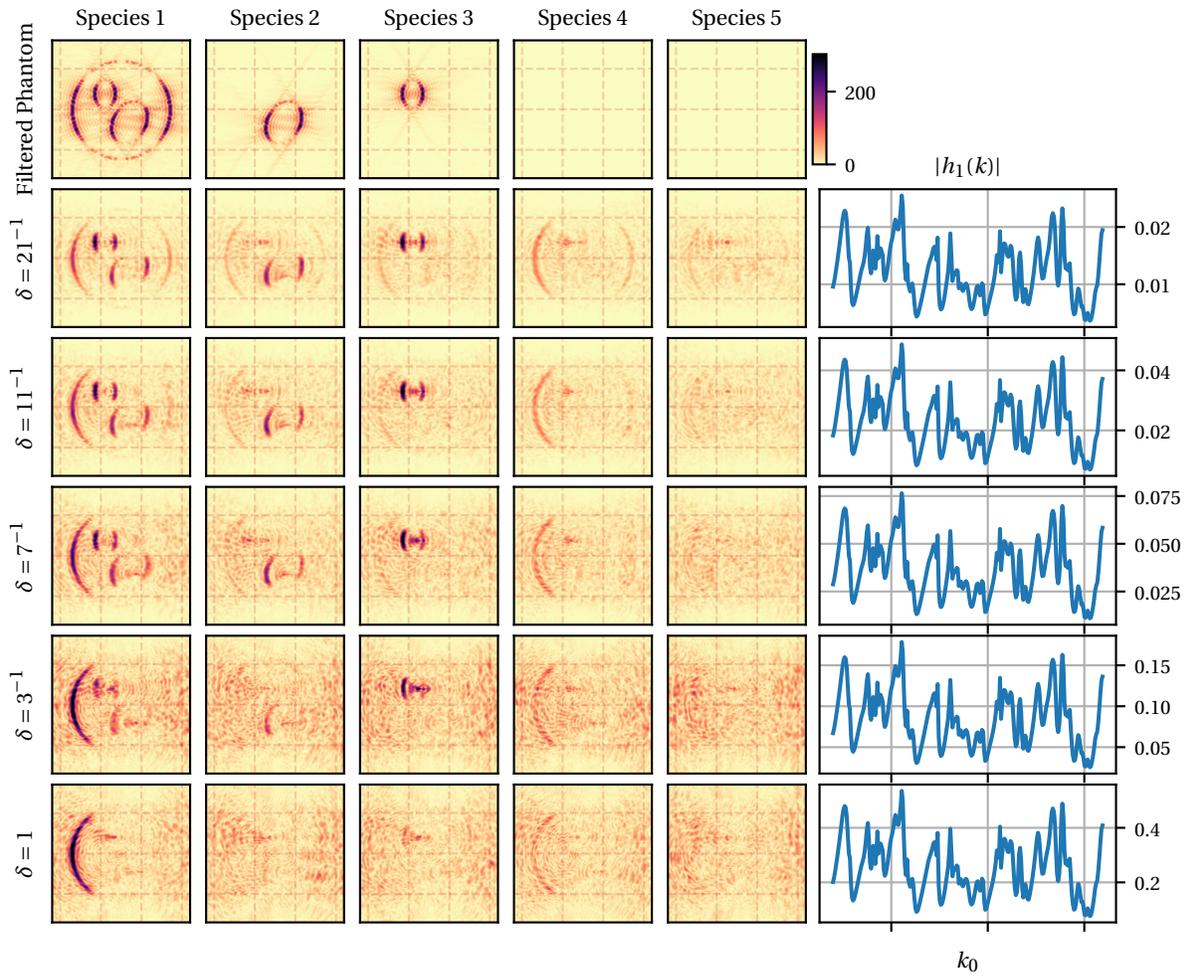}
  \caption{
    Reconstructions of cell phantom as a function of scattering strength.  All reconstructions use
    Tikhonov regularization with $\lambda_r = 10^{-5}$.   
  }
  \label{fig:cell_by_species_noisy}
\end{figure}

\begin{figure}
  \centering
  \input{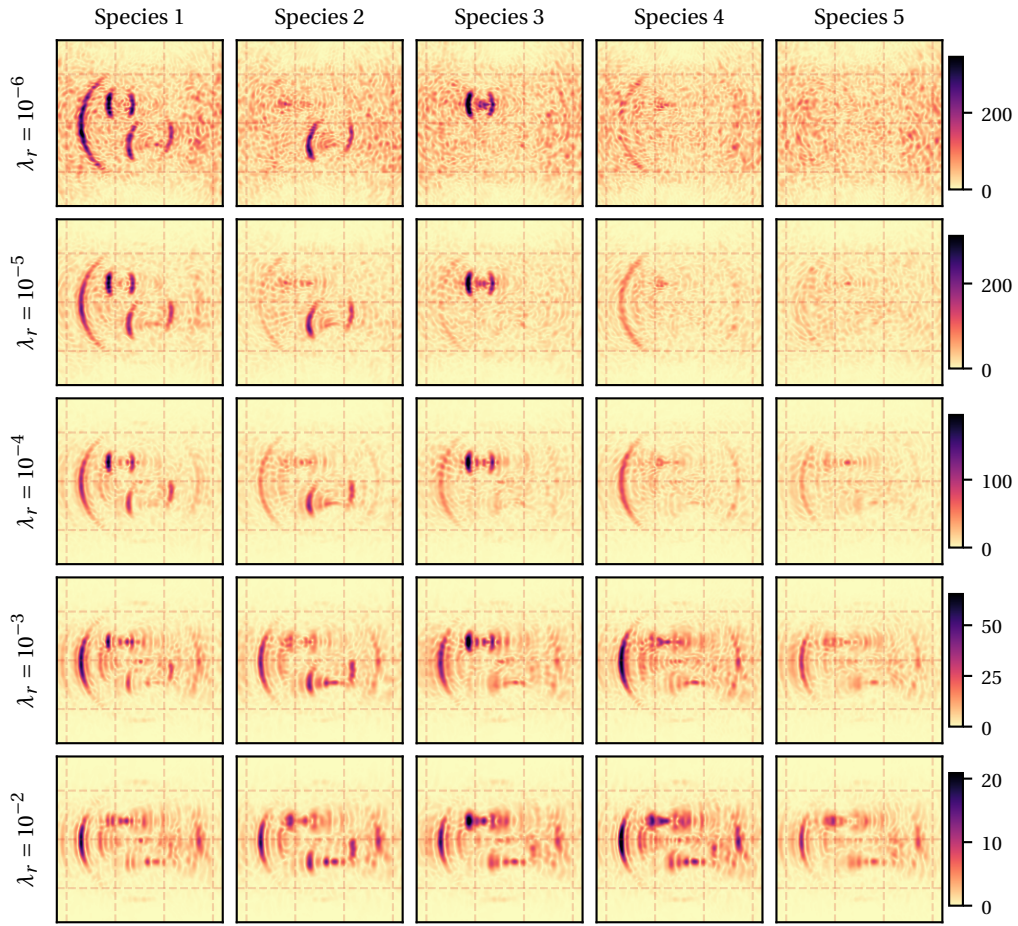}
  \caption{Reconstructions of cellular phantom using Tikhonov regularization and various levels of
    $\lambda_r$. The scattering strength parameter is $\delta=11^{-1}$.
    As $\lambda_r$ increases, the reconstruction fails to distinguish between chemical species.
}
  \label{fig:cell_tik}
\end{figure}

\section{Conclusions}
We have considered the problem of chemically specific and spatially resolved tomographic imaging
from interferometric measurements. We require the target to be the linear combination of a finite
number of distinct chemical species given data at a small number of \emph{en-face} focal planes.
We developed necessary and sufficient conditions for unique recovery of a target satisfying this
model. Linear independence of the chemical spectra is not sufficient---additional spectral
diversity is required.

In this paper, we assume the chemical spectra were either known or drawn from a library of possible
spectra. In the latter case, the number of required focal planes scales with the number of chemicals
present in the sample, not the total number in the library. Future work will consider extension
fully blind problem.

Our approach requires interferometric (phase-resolved) measurements and solves the linearized
scattering problem. This extension to intensity-only measurements and the removal of the Born
approximation are two avenues for future work.

Phaseless, intensity-only diffraction tomography has been demonstrated by modifying the acquisition
scheme \cite{Anastasio2006, Anastasio2009, Horstmeyer2016-Diffr-Tomog} and by optimization-based
approaches \cite{Pham2018-Versat-Recon}. Advances in high performance computing
\cite{Meng2018-Wideb-Fast, Hidayetoglu2018-Super-Full, Hidayetoglu2018-Fast-Massiv} and deep
learning \cite{Liu2017a, Sun2018-Effic-Accur, Soubies2017} have facilitated the solution of large
scale inverse scattering problems without linearization. In some cases, solving the nonlinear
inverse scattering problem overcomes the ``missing cone'' effect that hampers reconstruction of
extended targets. However, thus far, these approaches have only considered non-dispersive objects.
Extension of these methods to spectroscopic tomography within the \nsp approximation is an exciting
area of future work.

\section*{Acknowledgments}
The work of L. Pfister was partially supported by the Andrew T. Yang fellowship.
This research is based upon work partially supported by the Office of the Director of National
Intelligence (ODNI), Intelligence Advanced Research Projects Activity (IARPA), via Air Force
Research Laboratory (AFRL) FA8650-16-C- 9108. The views and conclusions contained herein are those
of the authors and should not be interpreted as necessarily representing the official policies or
endorsements, either expressed or implied, of the ODNI, IARPA, AFRL or the U.S. Government. The U.S.
Government is authorized to reproduce and distribute reprints for Governmental purposes
notwithstanding any copyright annotation thereon.

\appendix
\section{Proof of Main Theorems}
\label{sec:proofs}
\begin{proof}[Proof of \cref{thm:rank_uniqueness}]
  \cref{cond:unique} $\implies$ \cref{cond:nullspace}: Let $\pbl \in \nsetbarperp$ be the
  unique solution to $\sbl = \Phil \pbl$. Let $\mathbf{x} \in
  \nullspace{\Phil} \cap \nsetbarperp$. Now $\Phil(\pbl + \mathbf{x}) = \Phil
  \pbl = \sbl$. As $\mathbf{x} + \pbl \in \nsetbarperp$, by
  \cref{cond:unique} $\mathbf{x} = 0$. Thus \cref{cond:unique} $\implies$ \cref{cond:nullspace}.

  \cref{cond:nullspace} $\implies$ \cref{cond:rank}: Recall $\Phitl = \Phil (\mathbf{I}_{\Nf}
  \otimes \Vl) \in \Cbb^{\Nf \Nk \times \Ns r}$. As $\mathbf{I}_{\Nf} \otimes \Vl$ is
  a basis for $\nsetbarperp$, and $\nullspace{\Phil} = \nsetbar$ by assumption, $\Phitl
  \mathbf{x} = \mathbf{0}$ if and only if $\mathbf{x} = \mathbf{0}$; thus $\nullspace{\Phitl}
  = \left\{ \mathbf{0} \right\}$. By the rank nullity theorem, $\rank{\Phitl} = \Ns r$.

  \cref{cond:rank} $\implies$ \cref{cond:unique}: Suppose $\exists \mathbf{u}, \mathbf{v} \in
  \nsetbarperp$ such that $\Phil \mathbf{u} = \Phil \mathbf{v}$.
  As $\mathbf{I}_{\Nf} \otimes \Vl$ is a basis for $\nsetbarperp$, there are unique vectors 
  $\mathbf{x}, \mathbf{y}$ such that $\mathbf{u} = (\mathbf{I}_{\Nf} \otimes \Vl) \mathbf{x}$
  and $\mathbf{v} = (\mathbf{I}_{\Nf} \otimes \Vl) \mathbf{y}$.  
  Now $\mathbf{0} = \Phil (\mathbf{u} - \mathbf{v}) = \Phitl(\mathbf{x} - \mathbf{y}) \implies
  \mathbf{x} = \mathbf{y}$ as $\Phitl$ is full column rank; thus $\mathbf{u} = \mathbf{v}$,
  completing the proof.
\end{proof}
The following lemma regarding the rank of the Khatri-Rao product will prove useful:
\begin{lemma}
  \label{thm:kr_rank}
  Given $\mathbf{A} \in \Cbb^{m \times n_1}$ and $\mathbf{B} \in \Cbb^{m \times n_2}$,
  $\rank{\mathbf{A} \odot \mathbf{B}} \leq \min{(m, \rank{\mathbf{A}}\rank{\mathbf{B}})}$.
\end{lemma}
\begin{proof}
  As $\mathbf{A}\odot \mathbf{B} \in \Cbb^{m \times n_1 n_2}$, we have $\rank{\mathbf{A}\odot
    \mathbf{B}} \leq \min{(m, n_1 n_2)}$. Note that $\mathbf{A} \odot \mathbf{B}$ contains a subset
  of rows of the matrix $\mathbf{A} \otimes \mathbf{B}$. As the rank of the Kronkecker product is
  equal to the product of the ranks of $\mathbf{A}$ and $\mathbf{B}$ (\eg, \cite{Golub2012}), we have
  $\rank{\mathbf{A} \odot \mathbf{B}} \leq \rank{\mathbf{A} \otimes \mathbf{B}} = \rank{\mathbf{A}}
  \rank{\mathbf{B}}$.
\end{proof}

\begin{proof}[Proof of \cref{thm:n_species_necc}]
  Here, we suppress the superscript $\qp$.
  By \cref{thm:rank_uniqueness}, it suffices to show that the proposed conditions are necessary for
  $\tilde{\Phi}$ to have rank $\Ns r$.
  \cref{nec:dimensions} follows as $\tilde{\Phi}$ can have rank $\Ns r $
  only if $\Nk \Nf \geq \Ns r$.

  We show \cref{nec:independent} by
  contradiction; suppose $\rank{\mathbf{H}} = q < \Ns$. By construction
  $\rank{\Hbar} = \rank{\mathbf{H}}$. Thus by \cref{thm:kr_rank},
  $\rank{\tilde{\Phi}} \leq \rank{\Hbar}\rank{\Abarl} \leq r q < \Ns r$.

  For \cref{nec:orthogonality}, suppose the first row of $\B$ is orthogonal to the remaining
  $\Nk\Nf$ rows. Let $\mathbf{x}$ be a column vector formed from first row of $\B$
  and let $\mathbf{e}_1 \triangleq [1, 0, \hdots, 0] \in \Cbb^{\Nk\Nf}$; by construction,
  $\B \mathbf{x} = \mathbf{e}_1$.
  Set $\alpha = \sum_{\ns=2}^{\Ns} \mathbf{h}_{\ns}[1] / \mathbf{h}_1[1]$; then
  \begin{equation}
    \tilde{\Phi} \
    [ -\alpha \mathbf{x}^{\mathsf{T}}, \mathbf{x}^{\mathsf{T}}, \hdots, \mathbf{x}^{\mathsf{T}}]^{\mathsf{T}} =
    \mathrm{diag}\left\{\sum_{\ns=2}^{\Ns} \mathbf{h}_{\ns} - \alpha \mathbf{h}_1\right\} \mathbf{e}_1 = \mathbf{0},
  \end{equation}
  and so $\rank{\Phi} \leq \Ns r -1$.

  To show \cref{nec:spec_diversity},
  suppose there is a subset $J$ with $\abs{J} \geq \Ns$ such that
  $\mathbf{H}[J, \colon] \in
  \Cbb^{\abs{J} \times \Ns}$ is rank $\Ns$ and the remaining rows, $\mathbf{H}[J^c, \colon] \in \Cbb^{\Nk -
    \abs{J} \times \Ns}$ has rank $q < \Ns$.
  Define $\tilde{\Phi}^J \in \Cbb^{\Nf \abs{J} \times \Ns r}$
  to be the rows of $\tilde{\Phi}$ involving the rows of $\mathbf{H}$ indexed by $J$; that is,
  \begin{equation}
    \tilde{\Phi}^J =
    \begin{bmatrix}
      \mathbf{H}[J, \colon] \odot \B_1[J, \colon] \\
      \vdots \\
      \mathbf{H}[J, \colon] \odot \B_{\Nf}[J, \colon]
    \end{bmatrix},
  \end{equation}
  and construct $\tilde{\Phi}^{J^c} \in \Cbb^{\Nf (\Nk - \abs{J}) \times \Ns r}$ using the rows indexed
  by $J^c$.  
  As both $\B[J, \colon] \in \Cbb^{\Nf \abs{J} \times r}$ and $\B[J^c, \colon] \in \Cbb^{\Nf (\Nk -
    \abs{J}) \times r}$ have rank at most $r$, by
  \cref{thm:kr_rank}, we have
  \begin{equation}
    \rank{\tilde{\Phi}} \leq 
    \rank{\tilde{\Phi}^J} + \rank{\tilde{\Phi}^{J^c}} 
    \leq
  \min{\left(\Nf \abs{J}, \Ns r\right)} + 
  \min{\left(\Nf (\Nk - \abs{J}) , q r\right)} \triangleq \beta.
  \label{eq:gamma_ub}
  \end{equation}
  We wish to establish conditions such that $\beta \geq \Ns r$. This is clearly true, regardless
  of $q$, when $\Nf \abs{J} \geq \Ns r$. When $\abs{J} < \Ns r / \Nf$, we have
  \begin{equation}
    \beta = \Nf \abs{J} + \min{\left(\Nf (\Nk - \abs{J}) , q r\right)} .
  \end{equation}
  Suppose $\Nf(\Nk - \abs{J}) < q r$; then $\beta = \Nf \Nk \geq \Ns r$
  where the inequality follows from condition \cref{nec:dimensions}. Otherwise, if $\Nf (\Nk - \abs{J}) \geq q r$,
  then $\beta = \Nf \abs{J} + q r$ and $q \geq \Ns - \Nf \abs{J} / r$ implies $\beta \geq \Ns r$.

  To show \cref{nec:meas_diversity}, for each $i \in [\Nk]$ we define the index set 
  $J_i = \left\{ i, i + \Nk, \hdots, i + (\Nf-1) \Nk \right\}$; now, 
  \begin{equation}
    \tilde{\Phi}^{J_i} = (\ones{\Nf}^{\mathsf{T}} \otimes \mathbf{\Hbar}[J_i, \colon]) \odot \Bbar[J_i, \colon] = 
    \begin{bmatrix}
      \mathbf{h}_1[i] \B_1[i, \colon] &
      \mathbf{h}_2[i] \B_1[i, \colon] &
      \hdots & \mathbf{h}_{\Ns}[i] \B_1[i, \colon]   \\
      \vdots   & \vdots & & \vdots \\ 
      \mathbf{h}_1[i] \B_{\Nf}[i, \colon] &
      \mathbf{h}_2[i] \B_{\Nf}[i, \colon] &
      \hdots &  \mathbf{h}_{\Ns}[i] \B_{\Nf}[i, \colon] 
    \end{bmatrix}
    \in \Cbb^{\Nf \times \Ns r}.
  \end{equation}
  Now,
  $\rank{\tilde{\Phi}} \leq \sum_{i=1}^{\Nk} \rank{\tilde{\Phi}^{J_i}}
  \leq \sum_{i=1}^{\Nk} \rank{\B[J_i,\colon]}$, where the final inequality follows from
  \cref{thm:kr_rank}
  and $\rank{(\ones{\Nf}^{\mathsf{T}} \otimes \mathbf{H}[J_i, \colon])} = 1$.  Setting this upper bound to $\Ns r$
  gives the statement.
\end{proof}

\begin{proof}[Proof of \cref{thm:n_species_suff}]
  We omit the superscript $\mathbf{q}$.  It suffices to prove the case where $\tilde{\Phi}$ is square, $\Nk =
  r$ and $\Nf = \Ns$. Then $\rank{\tilde{\Phi}} \in \Cbb^{\Ns r \times \Ns r} = \Ns r$ if and only
  if
  \begin{align}
 \theta(\mathbf{H}) &\triangleq \det{\tilde{\Phi}} 
                      = \det{\mathbf{[\bar{H}} \odot \Bbar]}
       \neq 0.
  \end{align} 
  Now, $\theta(\mathbf{H})$ is a multivariate polynomial in the entries of $\mathbf{H}$ whose
  coefficents depend only on the entries of $\Bbar$. Thus $\theta(\mathbf{H})$ is
  either identically zero or its zero set is an affine algebraic set and thus a nowhere dense set of
  measure zero. It suffices to show $\theta(\mathbf{H}) \neq 0$ for a single choice of
  $\mathbf{H}$ (see, \eg, \cite{Harikumar1998-Fir-Perfec, Harikumar1999-Perfec-Blind,
    Jiang2001-Almos-Sure} and references therein).

We can permute the rows of $\tilde{\Phi}$ such that the first $\Nk$ rows are indexed by $J_1$, the
next $\Nk$ rows by $J_2$, and so on.  In particular, there is a permutation matrix $\Pi \in \Cbb^{\Nk \Ns
  \times \Nk \Ns}$ such that (\cf \cref{eq:phi_tilde_block})
\begin{align}
  \Pi \tilde{\Phi} = 
  \begin{bmatrix}
    \mathbf{D}_1[J_1, J_1] \B_1[J_1, \colon] &  \hdots & \mathbf{D}_{\Nf}[J_1, J_1] \B_1[J_1, \colon] \\
    \vdots                     &  & \vdots \\
    \mathbf{D}_1[J_1, J_1] \B_{\Nf}[J_1, \colon] &  \hdots & \mathbf{D}_{\Nf}[J_1, J_1] \B_{\Nf}[J_1, \colon] \\
    \mathbf{D}_1[J_2, J_2] \B_1[J_2, \colon] &  \hdots & \mathbf{D}_{\Nf}[J_2, J_2] \B_1[J_2, \colon] \\
    \vdots                     &  \ddots & \vdots \\
    \mathbf{D}_1[J_{\Nf}, J_{\Nf}] \B_{\Nf}[J_{\Nf}, \colon] &  \hdots & \mathbf{D}_{\Nf}[J_{\Nf}, J_{\Nf}] \B_{\Nf}[J_{\Nf}, \colon] 
  \end{bmatrix}
                                                                    = 
  \begin{bmatrix}
\mathbf{\check{D}}_{1}^{J_1} \mathbf{C}_1 &  \hdots & \mathbf{\check{D}}_{\Nf}^{J_1} \mathbf{C}_1 \\
\mathbf{\check{D}}_{1}^{J_2} \mathbf{C}_2 &  \hdots & \mathbf{\check{D}}_{\Nf}^{J_2} \mathbf{C}_2 \\
    \vdots                     &  \ddots & \vdots \\
\mathbf{\check{D}}_{1}^{J_{\Nf}} \mathbf{C}_{\Nf} &  \hdots & \mathbf{\check{D}}_{\Nf}^{J_{\Nf}} \mathbf{C}_{\Nf}
  \end{bmatrix},
\end{align}
where, in an abuse of notation, we write $\mathbf{\check{D}}_i^{J} = (\ones{\Nf}^{\mathsf{T}} \otimes
\mathbf{D}_i[J, J])$.

Next, we specify the choice of $\mathbf{H}$.
By assumption, $r = \Nk = m \Nf$ for some integer $m$. For each $i \in [\Ns]$, we set $\mathbf{h}_{i}[J_i] =
\mathbf{1}_m$ and set remaining coordinates are set to zero. With this construction,
$\mathbf{\check{D}}_{i}^{J_j} = \mathbf{I}_{m}$ if $i = j$; otherwise, $\mathbf{\check{D}}_{i}^{J_j} =
\mathbf{0}_{m}$.  Now
  \begin{align}
    \Pi \tilde{\Phi} = 
                        \begin{bmatrix}
                          \mathbf{C}_1 & \mathbf{0}_m & \hdots & \mathbf{0}_m \\
                          \mathbf{0}_m & \mathbf{C}_2& \hdots & \mathbf{0}_m \\
                          \vdots & \vdots  & \ddots & \vdots \\
                          \mathbf{0}_m & \mathbf{0}_m & \hdots & \mathbf{C}_{\Nf}
                        \end{bmatrix}, 
  \end{align}
  that is,  $\tilde{\Phi}$ is similar to a block diagonal matrix.  As each block along the diagonal is
  full rank by assumption, $\tilde{\Phi}$ is full rank.
\end{proof}

\bibliographystyle{myIEEEtran}
\bibliography{IEEEabrv,/home/pfister/research/bib/jabref}

\end{document}